\documentclass[11pt, a4paper]{article} 
\usepackage[affil-it]{authblk}
\usepackage{hyperref}
\usepackage{amssymb, amsmath, bm, mathrsfs, amsthm}
\usepackage[T1]{fontenc}
\usepackage{wallpaper}
\usepackage{graphicx, xfrac}
\usepackage[margin=2.5cm]{geometry}
\usepackage{times}
\usepackage{hyperref}
\hypersetup{
    colorlinks,
    citecolor=darkgray,
    filecolor=darkgray,
    linkcolor=darkgray,
    urlcolor =darkgray
}
\usepackage{pgf,tikz}
\usetikzlibrary{arrows}
\newcommand{\bu}{\bm{u}}
\newcommand{\bnu}{\bm{\nu}}
\newcommand{\bpi}{\bm{\pi}}
\newcommand{\bx}{\bm{x}}
\newcommand{\demi}{\frac{1}{2}}
\newcommand{\bS}{\bm{S}}
\newcommand{\bv}{\bm{v}}
\newcommand{\bw}{\bm{w}}
\newcommand{\br}{\bm{r}}
\newcommand{\ba}{\bm{a}}
\newtheorem{remark}{Remark}
\newtheorem{proposition}{Proposition}

\title{Efficient solvers for shallow-water Saint-Venant equations and
debris transportation-deposition models 
}
\author{Florian De Vuyst \thanks{email: fdevuyst@utc.fr}}
\affil{Universit\'e de Technologie de Compi\`egne, Alliance Sorbonne Universit\'e, Laboratoire de Math\'ematiques Appliqu\'ees de Compi\`egne (LMAC), Compi\`egne, France.}
\date{15 February 2021}
\begin{document}
\maketitle
\begin{abstract}
This research is aimed at achieving an efficient
digital infrastructure for evaluating risks and damages caused by tsunami flooding.
This research has been mainly focused on the suitable
modeling of debris dynamics for a simple (but accurate enough) assessment of damages. For different reasons including computational performance and Big Data management issues, we focus our research on Eulerian debris flow modeling. 
Rather than using complex multiphase debris models, we rather
use an empirical transportation and deposition model that takes into account
the interaction with the main water flow, friction/contact with the ground
but also debris interaction.
In particular, for debris interaction, we have used ideas coming from
vehicular traffic flow modeling. We introduce a velocity regularization
term similar to the so-called ``anticipation term'' in traffic flow modeling
that takes into account the local flow between neighboring debris and makes
the problem mathematically well-posed. It prevents from the generation of ``Dirac
measures of debris'' at shock waves.
As a result, the model is able to capture emerging phenomenons like debris aggregation and accumulations, and possibly to react on the main flow
by creating hills of debris and make the main stream deviate. 
We also discuss the way to derive quantities of interest (QoI), especially ``damage functions'' from the debris density and momentum fields. 
We believe that this original unexplored debris approach can lead to 
a valuable analysis of tsunami flooding damage assessment with 
Physics-based damage functions.
Numerical experiments show the nice behaviour of the numerical solvers, including the solution of  Saint-Venant's shallow water equations and debris dynamics equations. \\ [2ex]

\end{abstract}
\newpage
%
%
\section{Executive summary} 
%
%
This work is dedicated to the construction of an efficient numerical infrastructure for evaluating risks and damages caused by tsunami flooding. 

In this document, we mainly focus on the suitable
modeling of debris dynamics and the derivation of damage functions. For different reasons including computational performance and Big Data management issues, we focus our research on Eulerian debris flow modeling. 
Rather than using complex multiphase debris models (see for example~\cite{He2014,Pudasaini2012,Pudasaini2005,Bouchut2013}) or complex non-Newtonian models, we rather
use an empirical transportation and deposition model that takes into account
the interaction with the main water flow, friction/contact with the ground
but also debris interaction.
In particular, for debris interaction, we have used ideas coming from
vehicular traffic flow modeling (\cite{awrascle,zhang}). We introduce a velocity regularization
term similar to the so-called ``anticipation term'' in traffic flow modeling
that takes into account the local flow between neighboring debris and makes
the problem mathematically well-posed. It prevents from the generation of ``Dirac
measures of debris'' at shock waves.
As a result, the model is able to capture emerging phenomenons like debris aggregation and accumulations, and possibly to react on the main flow
by creating hills of debris and make the main stream deviate.  \medskip

We also discuss the way to derive quantities of interest (QoI), especially ``damage functions'' from the debris density and momentum fields. 
We believe that this original unexplored debris approach can lead to 
a valuable analysis of tsunami flooding damage assessment with 
Physics-based damage functions.
Numerical experiments show the nice behavior of the numerical solvers, including the solution of  Saint-Venant's shallow water equations and debris dynamics equations. 
\medskip

The next step is to introduce this debris model into a high-performance code like
VOLNA-OP2~\cite{dutykh,giles,giles2020}, then use a computer design of experiment (DoCE) with damage analysis,
sensitivity analysis and uncertainty quantification. There is clearly a 
Big Data issue in this computational context and we will try to propose original numerical methodologies but also suitable parallel software environments to process
the data and extract knowledge in this risk assessment framework.

\section{Motivation and introduction} 
%
%
Among the natural disasters, tsunamis generated by earthquakes or landslides
may cause the death of thousands of people and damage important
urban infrastructures. Recent events like the 2004 Indonesia and 2011 Japan
tsunamis have caused dramatic damages by severe flooding (see the figure~\ref{fig:flooding} below).
\begin{figure}[h!]
\begin{center}
\includegraphics[width=0.6\textwidth]{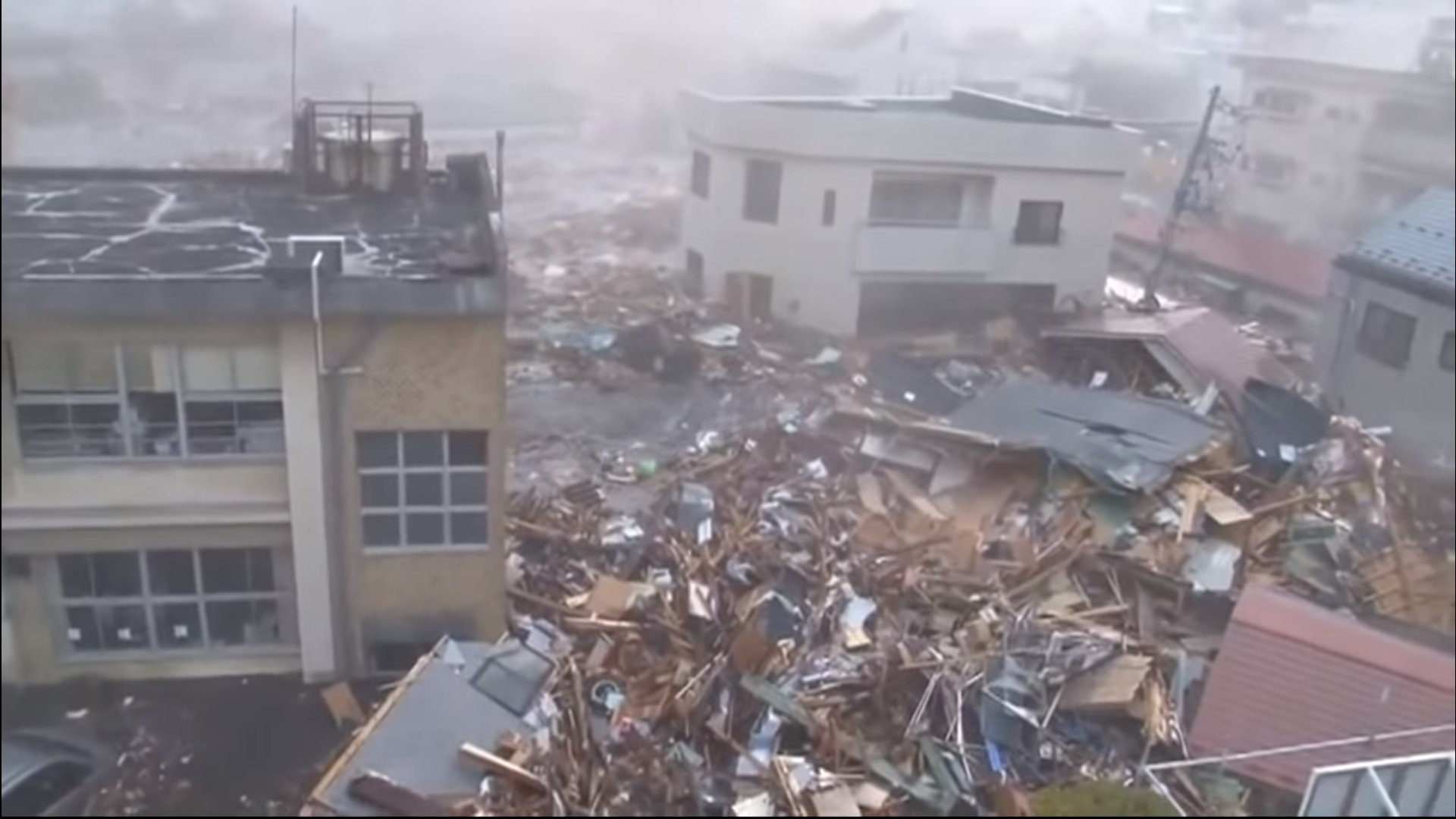}
\end{center}
\caption{Case of debris accumulation during the 2011 Japan tsunami flooding
(courtesy Youtube video).}\label{fig:flooding}
\end{figure}

More than a pure flooding, tsunami tidal waves drags millions of debris of all kinds:
cars, tree trucks, but also building materials, boats, etc. It is important
to prevent coastal population and cities from these major risks and to provide
forecast tools for that. Among the means of forecast, numerical simulation
appears to be powerful tool to assess damages and costs, achieve sensitivity analysis and uncertainty quantification~(UQ) and so on.
 
Acknowledged tsunami and general flooding computational codes are based on
the so-called shallow water equations, or Saint-Venant equations. They are
able to return rather good estimates of the tsunami wave propagation. The 
GPU-VOLNA code for example is a GPU-accelerated implementation
of the VOLNA code initially developed at CMLA ENS CACHAN by Dutykh et al~\cite{dutykh}.
It has be ported on large-scale GPU supercomputers by means of the OP2~platform~\cite{giles} that allow large-scale computations and designs of computer
experiments (DoCE) for statistical analysis. Runup elevation and
runup distances are the commonly damages estimates computed from shallow water
codes, and DoCE are able to return high-fidelity metamodels (or emulators)
of runup. There could be many other candidate damage functions like for
example instantaneous or accumulated flow rates that act as loadings on mechanical structures and as the main factor of fatigue breaking and generation of debris. 

However, Indonesia and Japan tsunami events have shown us that debris flows
are another important feature of tsunami flooding and, as such, have to be
taken into account into the general flooding flow. This need is originally
the reason of the present research.

One can observe that damage and risk assessment of tsunami flooding is a 
multicriteria/multiobjective problem subject to a high-dimensional
configuration space, where each evaluation requires heavy computations
of high-performance parallel supercomputers. There is clearly a
``Big Data'' dimension of the problem where we need to mine high-dimensional 
data (as computational results), extract data summaries, descriptors, 
indicators and risk/damage functions and provide emulators. The issue to perform 
this data analysis in a parallel efficient way is a current active field
of Research (map-reduce type algorithms, ...). Another issue is the way to perform
a design of computer experiment that can be done in an incremental manner 
in order to reduce metamodeling errors or uncertainty. This is also
a Big Data datawarehouse dimension that allows for easy analysis and 
visualization. \medskip

At our research developmental stage, this work is mainly dedicated to the
suitable modeling of debris dynamics that allows a rather good estimation of debris flow while being able to return damage functions. The models will also be
searched as simple as possible to allow ``fast'' evaluations and a full
exploration of the configuration space for statistical purposes. 
We also discuss the suitable design Saint-Venant solvers for tidal wave and
coast flooding applications where the capture of dry-wet phase transition
has to be done in an efficient and accurate manner. For that we introduce
a class of finite volume schemes we recently developed for multimaterial
compressible flow applications: the so-called Lagrange-flux schemes.  

%
%
\section{Introducing the Saint-Venant shallow-water equations}
%
Saint-Venant shallow water equations are a simplified two-dimensional model of a
truly three-dimensional flow of an incompressible fluid in contact with air and subject to gravity forces. These equations are mainly used for tsunami modeling as well as river flooding and for the analysis of dam break events.
In what follows, we will denote $g$ the gravity constant, $z=z(x,y)$ will
be the bathymetry ($z<0$) or landscape topography ($z>0$), $h=h(x,y)$ the
liquid depth and $\bu$ the vector field of depth-averaged velocity. 
In what follows, we will assume that $z$ is a Lipschitz continuous function.
Assuming that the $z$-dependency of the flow can be reduced to 
the knowledge of the water height, and supposing no ground drag forces, the mass and momentum balance equations read 
\begin{eqnarray}
&& \partial_t h + \nabla\cdot (h\bu) = 0, \label{eq:1} \\ [1.1ex]
&& \partial_t (h\bu) + \nabla\cdot(h \bu\otimes\bu) + \nabla p
 = -g h \nabla z, \label{eq:2}
\end{eqnarray}
where
\begin{equation}
p = p(h) = g \frac{h^2}{2} \label{eq:3}
\end{equation}
Here the notation ``$p$'' is used to emphasize the analogy with the usual
Euler equations of compressible gas dynamics, even it is not a pressure
from a thermodynamical point of view. 

It is well-known that this system of nonlinear partial differential equations
is hyperbolic. The quantity
\begin{equation}
c = \sqrt{\frac{dp}{dh}}=\sqrt{gh}	\label{eq:4}
\end{equation}
plays the role of a ``speed of sound'' in the system. For smooth solutions,
the energy quantity $\mathscr{E}$ defined by
\begin{equation}
\mathscr{E} = h \frac{|\bu|^2}{2} + g \frac{h^2}{2}.
\label{eq:5}
\end{equation}
is conserved and satisfies the additional balance equation
\[
\partial_t \mathscr{E} + \nabla\cdot (\mathscr{E}\bu) = - gh \nabla z \cdot \bu.
\]
As the energy is a convex function of the conservative variables $h$ and
$\bm{q}=h\bu$, it can also be considered as an entropy of the system. For general, 
possibly nonsmooth discontinuous solutions, we look for the physical 
(entropy) weak solution that fulfills the partial differential inequality
\begin{equation}
\partial_t \mathscr{E} + \nabla\cdot (\mathscr{E}\bu + p\bu) 
+ g h\nabla z\cdot \bu\leq 0
\label{eq:6}
\end{equation}
in the sense of distributions. \medskip

For numerical discretization, stable conservative entropy schemes are searched
in order to ensure convergence toward the entropy weak solution. In this
work, we propose to use the recent family of Lagrange-flux finite volume 
schemes \cite{devuyst} that provide both numerical stability and efficiency.
The construction of the Lagrange-flux is based on a Lagrangian remapping
process with a particular treatment of time discretization. This is detailed
in the next section.
\section{Lagrange-flux schemes for compressible Euler-type equations}
%
For complex compressible flows involving multiphysics phenomenons like e.g.\ high-speed elastoplasticity, multimaterial interaction, plasma, gas-particles, multiphase flows etc., a Lagrangian description of the flow is generally preferred for easier physical coupling. To ensure
robustness, some spatial remapping on a regular mesh may be added. A particular
case is the family of the so-called Lagrange+remap schemes~\cite{Hirt1974}, also referred to as Lagrangian remapping that apply a remap step on a reference (say Eulerian) mesh after each Lagrangian time advance step. Acknowledged legacy codes
implementing remapped Lagrange solvers usually define thermodynamical variables
at cell centers and velocity variables at mesh nodes (see figure~\ref{fig:1}).
\begin{figure}[h]
	\begin{center}
		\includegraphics[width=0.6\linewidth]{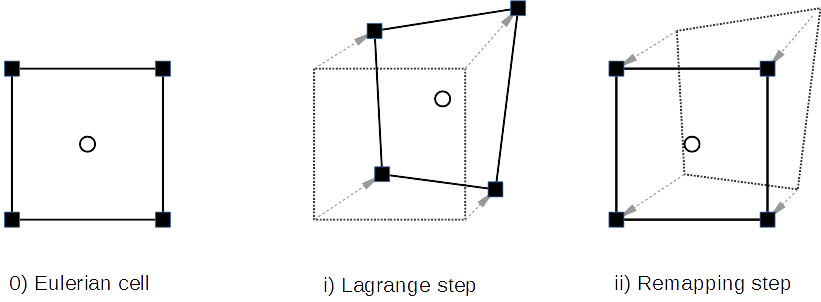}
		\caption{``Legacy'' staggered Lagrange-remap scheme: thermodynamical variables are located at cell centers (circles) whereas velocity variables are located at cell nodes (squares).}\label{fig:1}
	\end{center}
\end{figure}
\subsection{Performance issue, multicore/manycore computing}
In Poncet et al.~\cite{PARCO2015}, we have achieved a node-based performance analysis of a reference legacy Lagrange-remap hydrodynamics solver used in industry. By analyzing
each kernel of the whole algorithm, using roofline-type models \cite{Williams2009} on one side and
refined Execution Cache Memory (ECM) models  \cite{Treibig2010}, \cite{Stengel2015} on the other side, we have been able not only to quantitatively predict the performance of the whole algorithm --- with relative errors in the single digit range --- but also to identify a set of
features that limit the whole global performance. This can be roughly summarized into
three features:
\begin{enumerate}
	\item Staggered velocity variables involve a rather big amount of communication to/from CPU caches and memory with low arithmetic intensity, thus lowering the whole performance;
	\item Alternating direction (AD) strategies (see the appendix in~\cite{Collela1984}) or more specifically AD remapping procedures also generate too much communication with a loss of CPU occupancy and a rather poor multicore scalability.
	\item For multimaterial hydrodynamics using VOF-based interface reconstruction methods, 
	there is a strong loss of performance due to some array indirections and noncoalescent
	data in memory. Vectorization of such algorithms is also not trivial.\medskip
\end{enumerate}
From these observations and as a result of the analysis, we decided to ``rethink''
Lagrange-remap schemes, with possibly modifying some aspects of the solver in order
to improve node-based performance of the hydrocode solver. We have searched for
alternative formulations that lower communication and improve both arithmetic intensity
and SIMD property of the
algorithm. This redesign methodology has given us ideas of innovative Eulerian solvers. The so-called \emph{Lagrange-flux schemes} appear to be very promising in the Computational Fluid Dynamics (CFD) extended community, including
geophysical problems.  

Starting from a ``legacy'' staggered Lagrange-remap solver and related observed performance measurements, we want to improve the performance by modifying the computational approach under the following constraints and requirements:
\begin{enumerate}
	\item A Lagrangian solver must be used (for multiphysics coupling issue).
	\item To reduce communication, we prefer to use collocated cell-centered variables rather than
	a staggered scheme.
	\item To reduce communication, we prefer use a direct multidimensional remap solver rather than splitted alternating direction AD projections.
	\item The method should be simply extended to second-order accuracy (in both space and time). 
\end{enumerate}
Before going further, let us first comment the above requirements.
The second requirement should imply the use of a cell-centered Lagrange solver.
Fairly recently, Despr\'es and Mazeran in \cite{Despres2005} and 
Maire and et al.~\cite{Maire2007} 
have proposed pure cell-centered Lagrangian solvers based on 
the reconstruction of nodal velocities. In our study, we will examine if it is
possible to use approximate and simpler Lagrangian solvers in the Lagrange+remap context, in particular for the sake of performance. The fourth assertion
requires a full multidimensional remapping step, probably taking into
account geometric elements (deformation of cells and edges) if we want to ensure
high-order accuracy remapping. 
We have to find a good trade-off between simplifications-approximations and  accuracy (or properties) of the numerical solver.
\subsection{Lagrangian step of Lagrangian-remapping} \label{sec:3lag}
%
As example, let us consider the compressible Euler equations
for two-dimensional problems. Denoting $\rho,\ \bu=(u_i)_i,$ $i\in\{1,2\}$,\ $p$ and $E$ the density,
velocity, pressure and specific total energy respectively, the mass, momentum and
energy conservation equations are
\begin{equation}
\partial_t U_\ell + \nabla\cdot(\bu\, U_\ell) + \nabla\cdot\bpi_\ell=0,
\quad \ell=1,\dots,4,
\label{eq:7}
\end{equation}
where $U=(\rho,(\rho u_i)_i,\rho E)$, $\bpi_1=\vec 0$, $\bpi_2=(p,0)^T$,
$\bpi_3=(0,p)^T$ and $\bpi_4=p\bu$.
For the sake of simplicity, we will use a perfect gas equation of state
$p=(\gamma-1)\rho (E-\frac{1}{2}|\bu|^2)$, $\gamma\in(1,3]$. 
The speed of sound $c$ is given by $c=\sqrt{\gamma p/\rho}$.

For any volume $V_t$ that is advected by the fluid, from the Reynolds transport theorem we have
\[
\frac{d}{dt} \int_{\mathscr{V}_t} U_\ell\, d\bx =
\int_{\partial \mathscr{V}_t} 
\left\{\partial_t U_\ell+\nabla\cdot (\bu\, U_\ell)\right\}\, d\bx 
= - \int_{\partial \mathscr{V}_t} \bpi_\ell\cdot\bnu \, d\sigma
\]
where $\bnu$ is the normal unit vector exterior to $\mathscr{V}_t$. This leads
to a natural explicit finite volume scheme in the form
\begin{eqnarray}
|K^{n+1,L}| (U_\ell)_K^{n+1,L} = |K| (U_\ell)_K^n 
- \Delta t^n
\sum_{A^{n+\demi,L}\subset \partial K^{n+\demi,L}} |A^{n+\demi,L}|\, \bpi_A^{n+\demi,L}\cdot
\bnu_A^{n+\demi,L}.
\label{eq:8}
\end{eqnarray}
In expression~\eqref{eq:2}, the superscript ``L'' indicates the Lagrange evolution of the quantity.
Any Eulerian cell $K$ is deformed into the Lagrangian volume $K^{n+\demi,L}$ at time
$t^{n+\demi}$, and into the Lagrangian volume~$K^{n+1,L}$ at time $t^{n+1}$.
The pressure flux terms through the edges $A^{n+1/2,L}$ are evaluated at time~$t^{n+\demi}$ in order to get second-order accuracy in time. Of course, that means that we need a predictor
step for the velocity field $\bu^{n+\demi,L}$ at time $t^{n+\demi}$ (not written here
for simplicity).\smallskip

From now on, we will use the simplified notation $\bv^{n+\demi}=\bu^{n+\demi,L}$.
\subsection{Geometrical remapping step}\label{sec:3}
%
The remapping step consists in projecting the fields $U_\ell$ defined at cell centers
$K^{n+1,L}$ onto the initial (reference) Eulerian mesh with cells~$K$. Starting from an interpolated
vector-valued field~$\mathscr{I}^{n+1,L}U^{n+1,L}$, we project the field on 
piecewise-constant function on the Eulerian mesh, according to the integral formula
\begin{equation}
U_K^{n+1} = \frac{1}{|K|}\int_{K}  \mathscr{I}^{n+1,L}U^{n+1,L}(\bx)\,d\bx.
\label{eq:9}
\end{equation}
Practically, they are many ways to deal with the projection 
operation~\eqref{eq:9}. One can assemble elementary projection
contributions by computing the volume intersections between the 
reference mesh and the deformed mesh. But this procedure requires
the computation of \emph{all} the geometrical elements. Moreover, the projection
needs local tests of projection with conditional branching (think about the 
very different cases of compression, expansion, pure translation, etc). 
Thus the procedure 
is not SIMD and with potentially a loss of performance. 
The incremental remapping can also interpreted as a transport/advection
process, as already emphasized by Dukowicz and Baumgardner~\cite{Dukowicz2000}.
\subsection{Algebraic remapping}

Let us now write a different original formulation of the remapping process 
that does not explicitely requires the use of geometrical elements. 
In this step, there is no time evolution of any quantity, and in some sense
we have $\partial_t U=0$, that we rewrite
\[
\partial_t U = \partial_t U +\nabla\cdot(-\bv^{n+\demi} U)
\ +\ \nabla\cdot(\bv^{n+\demi} U) = 0.
\]
We decide to split up this equation into two substeps, a backward convection
and a forward one:
\begin{itemize}
	\item [i)] ~Backward convection: 
	\begin {equation}
	\partial_t U+ \nabla\cdot(-\bv^{n+\demi} U)=0.
	\label{eq:10}
\end{equation}
\item [ii)] ~Forward convection: 
\begin{equation}
\partial_t U+ \nabla\cdot(\bv^{n+\demi} U)=0.
\label{eq:11}
\end{equation}
\end{itemize}
Each convection problem is well-posed on the time interval $[0,\Delta t^n]$
under a standard CFL condition.
Let us now focus into these two steps and the way to solve them.
\subsubsection{Backward convection in Lagrangian description}
%
After the Lagrange step, if we solve the backward convection problem~\eqref{eq:4} 
over a time interval~$\Delta t^n$ using a Lagrangian description, we have
\begin{equation}
|K| (U_\ell)_K^{n,\star} =  |K^{n+1,L}| (U_\ell)_K^{n+1,L}. 
\label{eq:12}
\end{equation}
Actually, from the cell $K^{n+1,L}$ we go back to the original cell $K$ with
conservation of the conservative quantities. For $\ell=1$ (conservation of mass), 
we have
\[
|K|\, \rho_K^{n,\star} =  |K^{n+1,L}|\, \rho_K^{n+1,L}
\]
showing the variation of density by volume variation.
For $\ell=2,3,4$, it is easy to see that both velocity and specific total energy
are kept unchanged is this step:
\[
\bu^{n,\star} =  \bu^{n+1,L}, \quad 
E^{n,\star} =  E^{n+1,L}.
\]
Thus, this step is clearly computationally inexpensive.
\subsubsection{Forward convection in Eulerian description}
%
From the discrete field $(U_K^{n,\star})_K$ defined on the Eulerian cells $K$, 
we then solve the forward convection  problem~\label{eq:5} over a time step~$\Delta t^n$
under an Eulerian description. A standard Finite Volume discretization of the problem
will lead to the classical time advance scheme
\begin{equation}
U_K^{n+1} = U^{n,\star}_K - \frac{\Delta t^n}{|K|}\
\sum_{A\subset\partial_K} |A|\, U_A^{n+\demi,\star}\, (\bv_A^{n+\demi}\cdot \nu_A)
\label{eq:13}
\end{equation}
for some interface values $U_A^{n+\demi,\star}$ defined from the local neighbor values
$U^{n,\star}_K$. We finally get the expected Eulerian values $U_K^{n+1}$ at time $t^{n+1}$. \\[2ex]
Notice that from~\eqref{eq:6} and~\eqref{eq:7} we have also
\begin{equation}
|K|\,U_K^{n+1} = |K^{n+1,L}|\, U_K^{n+1,L} - \Delta t^n\
\sum_{A\subset\partial K} |A|\, U_A^{n+\demi,\star}\, (\bv_A^{n+\demi}\cdot \nu_A)
\label{eq:14}
\end{equation}
thus completely defining the remap step under the finite volume scheme form~\eqref{eq:14}.
Let us emphasize that we do not need any mesh intersection or geometric consideration to 
achieve the remapping process. The finite volume form~\eqref{eq:14} is now suitable
for a straightforward vectorized SIMD treatment. From~\eqref{eq:14} it is easy to achieve second-order accuracy for the remapping step by usual finite volume tools
(MUSCL reconstruction + second-order accurate time advance scheme for example).
\subsection{Full Lagrangian algebraic remapping time advance}
%
Let us note that the Lagrange+remap scheme is actually a conservative finite
volume scheme: putting~\eqref{eq:8} into~\eqref{eq:14} gives for all $\ell$:
\small
\begin{eqnarray}
(U_\ell)_K^{n+1} = (U_\ell)_K^n &-& \frac{\Delta t^n}{|K|}
\sum_{A^{n+\demi,L}\subset \partial K^{n+\demi,L}} |A^{n+\demi,L}|\, (\bpi_\ell)_A^{n+\demi,L}\cdot \nu_A^{n+\demi,L} 	\nonumber \\
&-& \frac{\Delta t^n}{|K|}\
\sum_{A\subset\partial K} |A|\, (U_\ell)_A^{n+\demi,\star}\, (\bv_A^{n+\demi}\cdot \nu_A)
\label{eq:15}
\end{eqnarray}
\normalsize
that can also be written
\small
\begin{eqnarray}
(U_\ell)_K^{n+1} = (U_\ell)_K^n &-& \frac{\Delta t^n}{|K|}
\sum_{A\subset\partial K}  |A|\,\left(\frac{|A^{n+\demi,L}|}{|A|}\, (\bpi_\ell)_A^{n+\demi,L}\cdot \nu_A^{n+\demi,L}\right) 	\nonumber \\
&-& \frac{\Delta t^n}{|K|}\
\sum_{A\subset\partial K}|A|\, \left( (U_\ell)_A^{n+\demi,\star}\, (\bv_A^{n+\demi}\cdot \nu_A)\right).
\label{eq:16}
\end{eqnarray}
\normalsize
We recognize into~\eqref{eq:16} pressure-related fluxes and convective 
numerical fluxes.
\subsection{Derivation of Lagrange-flux schemes}\label{sec:4}
%
From conclusions of the discussion above, we would like
to be free from any ``complex'' collocated Lagrangian solver involving complex geometric elements. Another difficult point is to define the deformation velocity field $\bv^{n+\demi}$ at time $t^{n+\demi}$, in an accurate manner. \medskip

In what follows, we are trying to deal with time accuracy in a different manner. 
Let us come back to the Lagrange+remap formula~\eqref{eq:16}. Let us consider
a ``small'' time step $t>0$ that fulfills the usual stability CFL condition
for explicit schemes. We have
\small
\begin{eqnarray*}
	(U_\ell)_K(t) = (U_\ell)_K^n &-& \frac{t}{|K|}
	\sum_{A\subset\partial K}  |A|\,\left(\frac{|A^L(t/2)|}{|A|}\, (\bpi_\ell)_A^{L}(t/2)\cdot \nu_A^{L}(t/2)\right) 	\nonumber \\
	&-& \frac{t}{|K|}\
	\sum_{A\subset\partial K}|A|\, (U_\ell)_A^{\star}(t/2)\, \bv_A(t/2)\cdot \nu_A.
\end{eqnarray*}
\normalsize
By making $t$ tend to zero, ($t>0$), we have $A^L(t/2)\rightarrow A$,
$(\bpi_\ell)^L(t/2)\rightarrow \bpi_\ell$,
$\bv(t/2)\rightarrow\bu$, $(U_\ell)^\star \rightarrow U_\ell$,
then we get a semi-discretization in space of the conservation laws. That can be seen as
a method-of-lines discretization (see~\cite{Schiesser1991}):
\small
\begin{equation}
\frac{d (U_\ell)_K}{dt} = - \frac{1}{|K|}
\sum_{A\subset\partial K}  |A|\, ((\bpi_\ell)_A\cdot \nu_A) 
-\frac{1}{|K|}\
\sum_{A\subset\partial K}|A|\, (U_\ell)_A\, (\bu_A\cdot\nu_A).
\label{eq:17}
\end{equation}
\normalsize
We get a classical finite volume method in the form
\[
\frac{dU_K}{dt} = - \frac{1}{|K|}\,\sum_{A\subset \partial K} |A| \ \Phi_A
\]
with a numerical flux $\Phi_A$ whose components are
\begin{equation}
(\Phi_\ell)_A =  (U_\ell)_A\, (\bu_A\cdot\nu_A) + (\bpi_\ell)_A\cdot \nu_A. 
\label{eq:18}
\end{equation}
In~\eqref{eq:17}, pressure fluxes $(\pi_\ell)_A$ and interface normal velocities $(\bu_A\cdot\bnu_A)$
can be computed from an approximate Riemann solver in Lagrangian coordinates
(for example the Lagrangian HLL solver, see Toro~\cite{Toro} for example). Then, the interface
states $(U_\ell)_A$ can be computed from a upwind process according to the sign
of the normal velocity $(\bu_A\cdot\nu_A)$. To get higher-order accuracy
in space, one can
use a standard MUSCL reconstruction + slope limiting process. At this stage,
because there is no time discretization, all acts on the Eulerian mesh and fluxes
are defined at the edges the the Eulerian cells. \medskip

To get high-order accuracy in time, one can then apply a 
standard high-order time advance scheme (Runge-Kutta~2, etc.). For the second-order Heun scheme
for example, we have the following algorithm:
\begin{enumerate}
	\item Compute the time step $\Delta t^n$ subject to some CFL condition;
	\item \textbf{Predictor step}. MUSCL reconstruction + slope limitation: from the discrete values~$U_K^n$,
	compute a discrete gradient for each cell $K$.
	\item Use a Lagrangian approximate Riemann solver to compute pressure fluxes $\bpi_A^n$ and interface velocities $\bu_A^n$
	\item Compute the upwind edge values $(U_\ell)_A^n$ according to the sign of
	$(\bu_A^n\cdot \nu_A)$;
	\item Compute the numerical flux $\Phi_A^n$ as defined in~\eqref{eq:12};
	\item Compute the first order predicted states $U_{K}^{\star,n+1}$:
	\[
	U_K^{\star,n+1} = U_K^n - \frac{\Delta t^n}{|K|}\sum_{A\subset \partial K}
	|A| \ \Phi_A^n
	\]
	\item \textbf{Corrector step}. MUSCL reconstruction + slope limitation: from the discrete values~$U_K^{\star,n+1}$,
	compute a discrete gradient for each cell $K$.
	\item Use a Lagrangian approximate Riemann solver to compute pressure fluxes $\bpi_A^{\star,n+1}$ and interface velocities $\bu_A^{\star,n+1}$
	\item Compute the upwind edge values $(U_\ell)_A^{\star,n+1}$
	according to the sign of $(\bu_A^{\star,n+1}\cdot \nu_A)$;
	\item Compute the numerical flux $\Phi_A^{\star,n+1}$ as defined in~\eqref{eq:12};
	\item Compute the second-order accurate states $U_{K}^{n+1}$ at time $t^{n+1}$:
	\[
	U_K^{n+1} = U_K^n - \frac{\Delta t^n}{|K|}\sum_{A\subset \partial K}
	|A|\ \frac{\Phi_A^n + \Phi_A^{\star,n+1}}{2}.
	\]
\end{enumerate}
One can appreciate the simplicity of the numerical solver.
\subsection{Details on the Lagrangian HLL approximate solver}
%
A HLL approximate Riemann solver \cite{Toro} in Lagrangian coordinates
can be used to easily
compute interface pressure and velocity. For a local Riemann problem made
of a left state~$U_L$ and a right state $U_R$, the contact pressure $p^\star$
is given by the formula
\begin{equation}
p^\star = \frac{\rho_R p_L + \rho_L p_R}{\rho_L+\rho_R}
- \frac{\rho_L\rho_R}{\rho_L+\rho_R}\, \max(c_L,c_R)\, (u_R-u_L),
\label{eq:19}
\end{equation}
and the normal contact velocity $u^\star$ by
\begin{equation}
u^\star =  \frac{\rho_L u_L + \rho_R u_R}{\rho_L+\rho_R}
- \frac{1}{\rho_L+\rho_R}\, \frac{p_R-p_L}{\max(c_L,c_R)}
\label{eq:20}
\end{equation}
leading to very simple operations.
\subsection{One dimensional numerical example}
%
\paragraph{Shock tube problems.} An as example, we test the Lagrange-flux scheme presented in section~\ref{sec:4} 
on the reference one-dimensional
Sod shock tube problem \cite{Sod71}. We use a Runge-Kutta 2 (RK2) time integrator and
a MUSCL reconstruction using the second-order Sweby slope limiter~\cite{Sweby1984} 
\[
\phi(a,b) = (ab>0) \ \text{sign}(a) \ \max\big( 
\min(|a|,\beta |b|),\ \min(\beta |a|, |b|)\big)
\]
with  coefficient $\beta=1.5$. We  use a uniform grid made of 384 points. 
The final time is $T=0.23$ and a CFL number equal to 0.25\,.
On figure~\ref{fig:sod}, one can observe a good behaviour of the
Eulerian solver, with rather sharp discontinuities and low numerical diffusion
into rarefaction fans.
\begin{figure*}
	\centering\includegraphics[height=0.16\textheight]{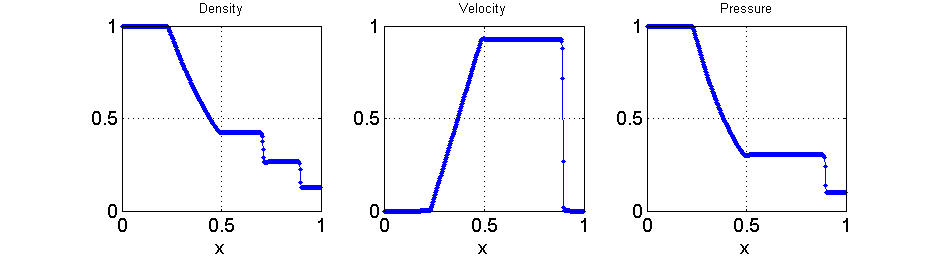}
	\caption{Second-order Lagrange-flux scheme on reference Sod's 1D shock tube problem. Time is $T=0.23$, 384 mesh points. Here Sweby's slope limiter
		with coefficient~1.5 is used.} \label{fig:sod}
\end{figure*}
\section{Derivation of Lagrange-flux scheme for the Saint-Venant equations}
%
This section is dedicated to the derivation of suitable Lagrange-flux
schemes for the Saint-Venant equations with gravity source terms.
\subsection{Requirements}
%
\subsubsection{Well-balanced property} 
For smooth solutions the momentum balance equation can be written
\[
\rho D_t\bu + \nabla\left(g(h+z)\right) = 0.
\]
A particular solution of the Saint Venant equations is the so-called ``lake-at-rest'' state, meaning that a fluid at rest ($\bu=0$ everywhere)
with topography $z$ satisfies the condition
\begin{equation}
h + z = C
\label{eq:21}
\end{equation}
for each connected part of the fluid domain. It is expected that numerical solver
also satisfies the fluid-at-rest condition at the discrete level, meaning that
\[
(h+z)_K = C
\]
for all finite volumes $K$ belonging to the same connected part of the discrete
domain~$\Omega^h$. This condition is name the ``well-balanced'' property. We are looking
for Lagrange-flux schemes that satisfy the well-balance property.
\subsubsection{Dealing with small water depth or wet-dry transitions}
For debris dynamics and flooding applications, it is particularly important to have a robust numerical scheme able to deal with small water depth and/or to handle
wet-dry state transitions. Rather than considering wet-dry interface reconstruction
techniques, for performance purpose we have decided use an interface capturing scheme that computes the variables in all the wet+dry domain. That means that
the numerical solver has to deal with vanishing water depth, see $h=0$ exactly. There
are three computational issues: first we have to ensure positivity of the water depth; Secondly, when $h$ vanishes, we have also to deal with vanishing
propagation speed $c=\sqrt{gh}$. Finally, because conservative variables are $h$ and $h\bu$ respectively, the fluid
velocity in ordinarily computed as
\[
\bu := \frac{h\bu}{h}.
\]
Of course this leads to ill-posed division operations for vanishing depths. Moreover, the division has no sense when $h=0$ exactly and something different
has to be done to compute the velocity.
\subsection{General form of the finite volume scheme}
For the sake of simplicity, only first-order accurate schemes are considered
in this document. We use standard finite volume notations with $K$ for a generic finite volume, $A$ for a generic edge, $\bnu_A$ a normal unit vector at edge $A$,
$\Delta t^n$ the time step at time $t^n$. 
For Saint Venant equations, we have looking for finite volumes explicit 
schemes in the form
\begin{eqnarray}
&& h_K^{n+1} = h_K^n - \frac{\Delta t^n}{|K|}\, \sum_{A\subset \partial K}
|A|\, \Phi_h(U_K^n, U_{K_A}^n,\bnu_A), 
\label{eq:22}\\ [1.1ex]
&& (h\bu)_K^{n+1} = (h\bu)_K^n - \frac{\Delta t^n}{|K|}\, \sum_{A\subset \partial K}
|A|\, \Phi_{hu}(U_K^n, U_{K_A}^n,\bnu_A)
- \frac{\Delta t^n}{|K|}\, \sum_{A\subset \partial K} |A|\, p_A^n \bnu_A
\nonumber \\ [1.1ex]
&& \phantom{(h\bu)_K^{n+1} = (h\bu)_K^n } 
- \frac{\Delta t^n}{|K|}\,  g\, \bar h_K^n \,\sum_{A\subset \partial K} |A|\, z_A \bnu_A,
\label{eq:23}
\end{eqnarray}
where $\Phi_h$ represents the numerical mass flux, $\Phi_{hu}$ represents
the convective flux related to the momentum variable, $p_A$ (resp. $z_A$) is the pressure flux (resp. topography value) at the edge $A$. In~\eqref{eq:23}, the last
term is a finite volume discretization of the source term $-\Delta t^n gh\nabla z$
into cell $K$.

To entirely define the numerical scheme, we have to characterize both convective
fluxes, pressure flux as well as the values for $z_A$ and $\bar h_K^n$.

The Lagrange-flux scheme methodology exposed above explains how to discretize convective fluxes and pressure fluxes as soon as an approximate Riemann solver is determined for the system. The choice for~$\bar h_K^n$ and~$z_A$ is guided by the
well-balanced property. This is developed in the next section. 
%
%
\subsection{Lagrangian approximate Riemann solver and well-balanced property}
%
For simplicity purposes, we consider in this section the one-dimensional Saint-Venant equations. We have first to achieve an approximate Riemann solver that
takes into account the gravity source term effect. Let us rewrite the momentum equation:
\[
\partial_t (hu) + \partial_x(h u^2) + \partial_x p = -gh \partial_x z
\]
with $p=p(h)=g \dfrac{h^2}{2}$ as usual. For a lake-at-rest solution, we have
\[
\partial_x p + hg \partial_x z = h(h+z)\partial_x (h+z)
-gz \partial_x (h+z) = \partial_x \left[g\frac{(h+z)^2}{2}\right]
-g\,z\, \partial_x (h+z) = 0.
\]
Let us now consider a Riemann problem made of two constant states 
$U_L=(h_L,(hu)_L)$ and $U_R=(h_R, (hu)_R)$ with a discontinuous topography
function with left and right values $z_L$ and $z_R$. 
In this section, we do not consider dry conditions, i.e. we assume that $h_L, h_R>0$.
Let us denote $z^\star$ a local mean topography value, function of both $z_L$ and $z_R$.
The replace the initial
momentum balance equation of the local conservation law
\begin{equation}
\partial_t (hu) + \partial_x (h u^2) + \partial_x \Pi = 0
\label{eq:24}
\end{equation}
with the pseudo pressure $\Pi$ defined by
\begin{equation}
\Pi = \Pi(h,z) = g \frac{(h+z)^2}{2} - g z^\star (h+z).
\label{eq:25}
\end{equation}
Let us emphasize that the modified system still has ``lake-at-rest'' solutions
with constant $\Pi$ giving the property $h+z=C$.

We retrieve a system of conservation laws that is similar to the isentropic Euler
equations with a pseudo pressure defined by~\eqref{eq:25}, and a propagation speed $\tilde c$
such that
\begin{equation}
\tilde c^2 = \frac{\partial \Pi}{\partial h} = gh + g (z-z^\star).
\label{eq:26}
\end{equation}
It is important to notice that $g (z-z^\star)$ has no sign, so this local approximate
problem has only a sense if involved depths are such that $h\geq z^\star-z$. The Lagrangian form of the equation writes
\begin{eqnarray*}
&& h D_t \tau - \partial_x u = 0, \\ [1.1ex]
&& h D_t u +\partial_x \Pi = 0
\end{eqnarray*}
with $\tau=h^{-1}$.
Introducing the mass variable $m$ such that $dm = h dx$, we have the conservative
script 
\begin{eqnarray*}
&& D_t \tau - \partial_m u = 0, \\ [1.1ex]
&& D_t u +\partial_m \Pi = 0.
\end{eqnarray*}
We now introduce an approximate Riemann solver made of two waves of respective
speeds $a_L$ and $a_R$ and intermediate state with depth $h^\star$ and velocity $u^\star$. By writing the Rankine-Hugoniot jump conditions on the second equation
for the two waves, we have the two relations
\begin{eqnarray}
&& a_L (u^\star-u_L) + \Pi_\star-\Pi_L = 0, \label{eq:27} \\ [1.1ex]
&& a_R (u_R-u^\star) + \Pi_R - \Pi^\star = 0 \label{eq:28}.
\end{eqnarray}
The ``acoustic solver'' approximation consider $a_L=-\sigma h_L$ and
$a_R = \sigma h_R$ for a propagation speed $\sigma$. The sub-characteristic
conditions requires that $\sigma\geq \max(c_L,c_R)$. From~\eqref{eq:27} 
and~\eqref{eq:28} we get the mean velocity
\begin{equation}
u^\star = \frac{h_L u_R + h_R u_R}{h_L + h_R} - \frac{1}{2\sigma}
\frac{2}{h_L+h_R} \left(\Pi_R-\Pi_L\right).
\label{eq:29}
\end{equation}
There is a convenient choice for $z^\star$ that gives an interesting 
and simple value for the difference of pseudo-pressure $\Pi_R-\Pi_L$:
\begin{proposition}\label{prop:1}
For $z^\star = \frac{1}{2}(z_L+z_R)$, we have
\[
\Pi_R - \Pi_L = g \frac{h_L+h_R}{2} \left[(h+z)_R-(h+z)_L\right].
\]
\end{proposition}
\begin{proof} It is easy to check that
\begin{eqnarray*}
\Pi_R-\Pi_L &=& p_R-p_R + g h_R(z_R-z^\star) - g h_L (z_L+z^\star)
+ g \left(\frac{z_L+z_R}{2}-z^\star\right) (z_R-z_L) \\ 
&=& p_R-p_L -g z^\star (h_R-h_L) + g((hz)_R-(hz)_L) + g \left(\frac{z_L+z_R}{2}- z^\star\right) (z_R-z_L)
\end{eqnarray*}
and, because $p_R-p_L = g\frac{h_L+h_R}{2}(h_R-h_L)$, one finds the result.
\end{proof}
Following the choice from proposition~\ref{prop:1}, we have the formula 
for $u^\star$:
\begin{equation}
u^\star = \frac{h_L u_R + h_R u_R}{h_L + h_R} - \frac{1}{2\sigma}
g\, \left[(h+z)_R-(h+z)_L\right].
\label{eq:30}
\end{equation}
We then have the trivial result:
\begin{proposition}\label{prop:2}
For lake-at-rest conditions, i.e. $u_L=u_R=0$ and $(h+z)_L=(h+z)_R$, we have $u^\star=0$.
\end{proposition}
As a direct consequence of Proposition~\ref{prop:2}, for lake-at-rest conditions, 
convective fluxes will be zero.

Let us go now to the computation of the mean depth $h^\star$. By writing the 
Rankine-Hugoniot jump conditions on the first equation for the two waves,
we have the relations
\begin{eqnarray*}
&& \phantom{-}h_L \sigma (\tau^\star - \tau_L) = u^\star-u_L, \\ [1.1ex]
&& -h_R \sigma (\tau_R -\tau^\star) = u_R - u^\star,
\end{eqnarray*}
thus giving
\[
\tau^\star = \frac{2}{h_L+h_R} + \frac{1}{2\sigma}\, \frac{2}{h_L+h_R}(u_R-u_L)
\]
or equivalently
\begin{equation}
h^\star = \frac{h_L+h_R}{2}\, \frac{1}{1+\frac{1}{2\sigma} (u_R-u_L)}.
\label{eq:31}
\end{equation}
Then the mean pressure $p^\star$ is simply computed as $p^\star = g\dfrac{(h^\star)^2}{2}$. We have the following well-balanced property~:
\begin{proposition}\label{prop:3}
Choosing $z^\star=\frac{1}{2}(z_L+z_R)$, for lake-at-rest conditions $u_L=u_R=0$
and $(h+z)_L=(h+z)_R$, we have $h^\star = \dfrac{h_L+h_R}{2}$ so that 
\[
h^\star+z^\star = (h+z)_L = (h+z)_R.
\]
\end{proposition}
\begin{remark}
In~\eqref{eq:31}, one can notice a ``compressibility'' term
\[
\kappa = 1 + \frac{1}{2\sigma}(u_R-u_L).
\]
For $u_R-u_L<0$ which expresses local compression conditions, in order
to keep $\kappa$ positive, we need to choose $\sigma$ such that
\[
\sigma > \frac{|u_R-u_L|}{2}.
\] 
On can choose for example $\sigma \leftarrow \max(\sigma, -\min(0, u_R-u_L))$.
\end{remark}
\subsection{Discretization of the source term and well-balanced property}
%
Now we discuss the discretization of the gravity source term. For one-dimensional
problems and a uniform mesh grid, the numerical scheme writes
\begin{equation}
(hu)_j^{n+1} = (hu)_j^n - \frac{\Delta t}{\Delta x} \left(\Phi_{hu,j+1/2}^n
-\Phi_{hu,j-1/2}^n \right) - \frac{\Delta t}{\Delta x} 
\left(p_{j+1/2}^n-p_{j-1/2}^n \right) - \frac{\Delta t}{\Delta x} g\bar h_j^n
\left( z_{j+1/2}-z_{j-1/2}\right),
\label{eq:32}
\end{equation}
where the $\Phi_{hu,j+1/2}^n$ are the momentum convective fluxes and the
$p_{j+1/2}^n$ are the pressure at cell interfaces, prescribed by the
approximate Riemann solver. In~\eqref{eq:32}, we have still to find a convenient
choice for both quantities $\bar h_j^n$ and interface topography $z_{j+1/2}$.
This choice is guided by the following result:
\begin{proposition}\label{prop:4}
Let us consider approximate Riemann values $u^\star$ and $h^\star$ given
by the formulas~\eqref{eq:30} and~\eqref{eq:31} respectively. Then, the
initial lake-in-rest conditions, i.e., $u_j=0$ and $(h+z)^0_j = c$ for all $j$,
then the choice
\begin{equation}
z_{j+1/2} = \frac{z_j+z_{j+1}}{2}, \quad 
\bar h_j^n = \frac{h_{j-1/2}^n + h_{j+1/2}^n}{2}
\label{eq:33}
\end{equation}
provides a well-balanced numerical scheme, i.e. $(h+z)_j^n = C$ for all $j$, 
for all $n\in\mathbb{N}$.
\end{proposition}
\begin{proof}
According to Proposition~\ref{prop:2}, the approximate Riemann solver returns
a null velocity value $u^\star$ under lake-in-rest conditions. Thus, convective
fluxes are zero. Let us assume that we have lake-in-rest conditions at time $t^n$.
Then we have
\begin{eqnarray*}
(hu)_j^{n+1} &=& - \frac{\Delta t}{\Delta x} 
\left(p_{j+1/2}^n-p_{j-1/2}^n -  g\bar h_j^n
\left( z_{j+1/2}-z_{j-1/2} \right)\right)  \\ [1.1ex]
\phantom{(hu)_j^{n+1}} &=& - \frac{\Delta t}{\Delta x} 
\left( g\left(\frac{(h_{j+1/2}^n)^2}{2}-\frac{(h_{j-1/2}^n)^2}{2} \right)-  g\bar h_j^n \left( z_{j+1/2}-z_{j-1/2} \right)\right) \\ [1.1ex]
\phantom{(hu)_j^{n+1}} &=& - \frac{\Delta t}{\Delta x} 
\left( g \frac{h_{j-1/2}^n + h_{j+1/2}^n}{2}\left(h_{j+1/2}^n-h_{j-1/2}^n \right)-  g\bar h_j^n \left( z_{j+1/2}-z_{j-1/2} \right)\right) 
\end{eqnarray*}
Then for $\bar h_j^n$ defined in~\eqref{eq:33}, we have
\[
(hu)_j^{n+1} = - \frac{\Delta t}{\Delta x} 
 g \frac{h_{j-1/2}^n + h_{j+1/2}^n}{2}
\left( (h_{j+1/2}^n+z_{j+1/2})-(h_{j-1/2}^n-z_{j-1/2})\right).
\]
According the Proposition~\ref{prop:3}, the right hand side is 0.
\end{proof}
\subsection{Dealing with dry-wet phase transitions}
%
As mentioned above, when $h$ vanishes, we have to deal with vanishing
propagation speed $c=\sqrt{gh}$ and by the fact that approximate Riemann 
problem quantities may leads to some divisions by zero, making the computational
method unstable. Moreover, the division
\begin{equation}
\bu := \frac{h\bu}{h}
\label{eq:34}
\end{equation}
required to compute the velocity may be arbitrary ill-conditioned for~$h$
close to zero. This generally produces spuriously large velocities at
wet-dry transition fronts, making the time step collapse or simply leading to 
a code crash. This section addresses these issues. \medskip

Let us first clarify some behavior of the solutions at vanishing water depth $h\rightarrow 0$ ($h\neq 0$). For smooth solutions, the momentum equation can be written
\[
h D_t \bu + g h\nabla (h+z) = 0.
\]
Dividing by $h$, we find the Burgers-like equation
\[
\partial_t \bu + \bu\cdot \nabla\bu + g \nabla(h+z) = 0.
\]
Neglecting $h$ in the equation gives
\begin{equation}
\partial_t \bu + \bu\cdot \nabla\bu = -g\nabla z.
\label{eq:35}
\end{equation}
We find out an autonomous equation in $\bu$. We want to insist on the fact that,
at vanishing $h$, the value of velocity is not arbitrary, but solution of
the above equation which is subject to gravity acceleration. In particular,
$\bu$ has not to be zero as often considered and encountered in the literature.
This is because friction terms are not taken into account so that fluid
is sliding on the ground and accelerating according to gravity forces.
Let us finally remark that equation~\eqref{eq:35} can be written in conservative form
by considering any function $\eta>0$  solution of the conservation law
\begin{equation}
\partial_t \eta + \nabla\cdot(\eta\bu) = 0.
\label{eq:36}
\end{equation}
Then we can rewrite~\eqref{eq:35} as
\begin{equation}
\partial_t(\eta \bu) + \nabla\cdot(\eta \bu\otimes\bu) = -g\eta \nabla z.
\label{eq:37}
\end{equation}
\subsubsection{Propagation speed and approximate Riemann solver}
%
Let us recall the interface velocity and water depth obtained with our Lagrangian
approximate Riemann solver, for ``wet-wet'' conditions:
\begin{eqnarray*}
&& u^\star = \frac{h_L u_R + h_R u_R}{h_L + h_R} - \frac{1}{2\sigma}
g\, \left[(h+z)_R-(h+z)_L\right], \\ [1.1ex]
&& h^\star = \frac{h_L+h_R}{2}\, \frac{1}{1+\frac{1}{2\sigma} (u_R-u_L)}
\end{eqnarray*}
with a sub-characteristic propagation speed $\sigma$ chosen for example as $\sigma = \max\left(c_L,c_R,-(u_R-u_L)_-\right)$. There are two difficulties: either
$(h_L+h_R)$ is small or $\sigma$ is small, what may lead to some ill-conditioned divisions. A way to fix the problem of ``speed of sound'' is to use a larger
non-vanishing propagation speed~$\sigma$. This can be theoretically justified by a relaxation technique~\cite{jin,suliciu}. \medskip

Let us remark that the expression for $u^\star$ still holds for $h_L h_R=0$ but
$h_L+h_R>0$. The remaining difficulty is the ``dry-dry'' case or vanishing $h$
on both two sides. In this case we have to use something different. As mentioned above, at least for smooth solutions the momentum equation can be rewritten
as
\[
D_t u + \partial_x(g(h+z)) = 0.
\]
Looking for an approximate Riemann solver considering this equation leads to
the mean velocity
\begin{equation}
u^\star = \frac{u_L+u_R}{2}- \frac{1}{2\sigma} g\, \left[(h+z)_R-(h+z)_L\right]
\label{eq:38}
\end{equation}
or
\begin{equation}
u^\star = \frac{u_L+u_R}{2}- \frac{1}{2\sigma} g\, \left(z_R-z_L\right)
\label{eq:39}
\end{equation}
if $h_L$ and $h_R$ are neglected. From the practical computational point of view, we have to
define a threshold on $h_L+h_R$ to use either~\eqref{eq:30} or~\eqref{eq:39} for 
computing~$u^\star$. Remark that expression~\eqref{eq:39} can be seen as a particular
case of~\eqref{eq:30} with $h_L=h_R=1$.
%
%
%
\subsubsection{Determination of the velocity vector}
%
The division operation $\bu:=(h\bu)/h$ cannot be used for pure dry conditions
($h=0$) or vanishing water depth $(h\ll 1)$. We have to do something different.
In this work, we propose to use an additional ``dry velocity'' equation. Let
us denote $\bu^{dry}$ the ``dry velocity''. It is solution of the transport-relaxation
equation
\begin{equation}
\partial_t \bu^{dry} + \bu^{dry}\cdot\nabla\bu^{dry} = 
-g\nabla z + \frac{\bu-\bu^{dry}}{\mu}
\label{eq:40}
\end{equation}
with $\mu>0$, $\mu\ll 1$ as relaxation coefficient. There are two strategies:
either we use a non-conservative scheme that discretizes the equation~\eqref{eq:40}, or we use a conservative scheme on the conservative form
\begin{equation}
\partial_t(\eta \bu^{dry})+\nabla\cdot(\eta\bu^{dry}\otimes\bu^{dry})
= -g\eta\nabla z + \eta\, \frac{\bu-\bu^{dry}}{\mu}.
\label{eq:41}
\end{equation}
using the additional variable $\eta$ solution of~\eqref{eq:36}. Let us comment is more details the numerical procedure. First, we have to define a field $\eta^n$ and
a field $\bu^{dry}$ at time $t^n$. Because $\eta^n$ is arbitrary, it is simply chosen as
\[
\eta^n_K = 1\quad \forall K\in\mathcal{T}^h.
\]
We project the initial dry velocity onto the actual velocity field $\bu^n$:
\[
(\bu^{dry})^n_K = \bu_K^n\quad \forall K\in\mathcal{T}^h.
\]
Then we apply the Lagrange-flux explicit scheme on the system
\begin{eqnarray*}
&& \partial_t \eta + \nabla\cdot(\eta\bu^{dry}) = 0, \\ [1.1ex]
&& \partial_t (\eta\bu^{dry}) +\nabla\cdot(\eta\bu^{dry}\otimes\bu^{dry}) 
= -g\eta\nabla z.
\end{eqnarray*}
For one-dimensional problems, this would give
\begin{eqnarray*}
&& \eta_{j}^{n+1} = 1 
- \frac{\Delta t}{\Delta x}\left(\Phi_{\eta,j+1/2}^n-\Phi_{\eta,j-1/2}^n\right), \\ [1.1ex]
&& \eta_j^{n+1}\, (u^{dry})_j^{n+1} = u_j^n - \frac{\Delta t}{\Delta x} \left(\Phi_{\eta u,j+1/2}^n
-\Phi_{\eta u,j-1/2}^n \right)  - \frac{\Delta t}{\Delta x} g
\left( z_{j+1/2}-z_{j-1/2}\right).
\end{eqnarray*}
%
%
\begin{remark}
For water depth vanishing conditions, the auxiliary variable $\eta$ can be interpreted 
as a water depth renormalization, where a scaling is operated in order to return
a rescaled depth of order $O(1)$. Because there is an uncertainty on the 
good scale to apply, we simply express the rescaling by $\eta^n=1$.
\end{remark}
Once the candidate velocity $\bu^{dry}$ in case of dry or near-dry conditions is computed, we have to define smoothed switching strategies to compute the
actual velocity. We proceed as follows:
for a given momentum value $(h\bu)$ and a given dry velocity $\bu^{dry}$, we search for a velocity vector $\bu^\star$ solution of the Tykhonov-regularized mean
square problem
\begin{equation}
(\mathcal{P}_\varepsilon)\quad\quad \min_{\bv}\quad 
\frac{1}{2} \| (h\bu) - h\, \bw\|^2 + \varepsilon\, \|\bw-\bu^{dry}\|^2
\label{eq:42}
\end{equation}
\begin{equation}
\bu^\star := \dfrac{h(h\bu) + \varepsilon\, \bu^{dry}}{h^2 + \varepsilon^2}.
\label{eq:43}
\end{equation}
Of course, for non vanishing $h$ and rather small $\varepsilon$, we have
$\bu^\star\approx \bu$ and for vanishing $h$, we get $\bu^\star\approx \bu^{dry}$.
This is similar to usual velocity fixes found in the literature, as in Kurganov and Petrova~\cite{kurganov} for example, but the authors actually 
consider~$\bu^{dry}=0$. The
choice $\bu^{dry}=0$ creates an artificial viscosity/damping into depth-vanishing
regions and in particular water cannot completely drain away.  
For the choice of $\varepsilon$, it can be adapted to the cell size as suggested
by Kurganov-Petrova.
\section{Empirical model of debris transportation and deposition}
%
This section is dedicated to the modeling of debris dynamics carried away by flooding or tsunami waves, including possible coalescence and deposition effects. For foreseen risk analysis purposes, the model has to be able to return high-level debris features and damage functions with sufficient fidelity. One must find a good trade-off between computational efficiency and ability to return quantities of interest like damage functions.
\subsection{Requirements}
%
Conditions of use of the debris models for risk analysis lead to the 
following requirements:
\begin{enumerate}
\item The computational debris models have to be rather efficient in terms of computational complexity, preferably at the order of the Saint-Venant numerical
solvers; \label{enu:1}
\item Moreover, the model complexity has to be independent of the number of debris, allowing for large-scale debris computations; \label{enu:2}
\item The model has to be able to return quantities of interest 
like damage functions for risk analysis.
\end{enumerate}
\subsection{Derivation of a model and model analysis}
%
For transport-dominated problems, we have the choice between Lagrangian or
Eulerian models. Lagrangian models track trajectories debris particles or group of particles and, because of that, are more accurate than Eulerian models. Unfortunately,
the algorithmic complexity is proportional to the number of debris particles and thus do not fulfill the requirement nb~2. Thus we rather move toward Eulerian
models, expressed in terms of averaged density quantities.

In the literature, one can find volume-averaged multiphase-based debris models
\cite{Pudasaini2012,Pudasaini2005,Bouchut2013,Hutter96,Iverson97}. But these models are known to be highly computationally expensive and thus
are irrelevant in our case. 

It is here proposed  to emulate well-known vehicle-based traffic flow for modeling debris dynamics, up to some adjustments.
\subsubsection{One-dimensional model}
%
For a one-dimensional liquid+debris flow, debris are driven by the main liquid flow
and are following themselves. Let us first consider a Lagrangian description of the flow. For a debris particle~$j$ of mass $m$, located at position $x_j(t)$ at time $t$ at velocity $v_j(t)$ and following the debris particle number $(j+1)$, we have the motion
equations
\begin{eqnarray}
&& \frac{d x_j}{dt} = v_j, \label{eq:44}\\ [1.1ex]
&& m\frac{d v_j}{dt} = - D_j(d) - F_j(f) + I_j(t),\label{eq:45}
\end{eqnarray}
where the term $D_j(t)$ represent a drag force due to the driving shallow water
flow, $F_j(t)$ is a ground friction term for debris in contact with ground and
$I_j(t)$ is a term that represents the interaction of debris nb $j$ with its
neighboring ones. These three terms have to be modelized and closed with respect
to the variables of computations.

For the drag term, one can simply use the relaxation term
\[
D_j(t) = m\,\frac{u(x_j(t))- v_j(t)}{\tau_D}
\]
where $\tau_D>0$ is a characteristic relaxation time. For the ground friction 
term, one can also use a relaxation term
\[
F_j(t) = -\omega_F \, m\, v_j(t),
\]
but the relaxation rate $\omega_F$ is a function of the water depth $h$. Let us introduce $h_f>0$
the characteristic debris plunge depth. It is expected that a debris rapidly
stops when $h<h_f$ by ground contact. On the other hand, when $h\gg h_f$, there
is no friction with ground and the characteristic relaxation time should be infinite.
We then propose to use the rather simple function 
\begin{equation}
\omega_F = \omega_F(h) = \frac{m}{\tau_F} \max\left(1,\ \frac{h_f}{h}\right)\,
\max\left(0, 1- \frac{h}{h_f}\right)^\beta,
\label{eq:46}  
\end{equation}
where $\beta>0$ and $\tau_F>0$ is a constant friction characteristic time. 

Debris interaction terms are needed to take into account unresolved small scale local water flow towards debris (rear debris recirculation, vortexes, 
suction, ...).  For one-dimensional problems, debris objects cannot overtake
themselves and the debris order is preserved. This can be expressed
by an acceleration/slowing down term as used in vehicular traffic flow \cite{awrascle,zhang}
(referred to as an anticipation term in car-following models). Here, we
empirically define the interaction term as a velocity relaxation between
neighboring debris: if $x_{j+1}(t)-x_j(t)>0$, for a flow moving to the right, one can define for example~$I_j(t)$ as
\begin{equation}
I_j(t) = m\, a\, \frac{v_{j+1}(t)-v_j(t)}{x_{j+1}(t)-x_j(t)}
\label{eq:47}
\end{equation}
where $a>0$ is homogeneous to a speed. 
The speed quantity $a$ can also be modeled. One could for example take it
proportional to the local debris velocity, i.e. $a=\lambda v_j(t)$, where 
$\lambda>0$ is a dimensionless constant. For a general debris, we then consider
following interaction term:
\begin{equation}
I_j(t) = \lambda\, m\, v_j(t)\, \frac{v_{j+1}(t)-v_j(t)}{x_{j+1}(t)-x_j(t)}
\label{eq:48}
\end{equation}
As a summary, we consider the following discrete dynamical system:
{\small
\begin{eqnarray}
&& \dot{x_j} = v_j, \label{eq:49}\\ [1.1ex]
&& \dot{v_j} = \frac{u(x_j(t))- v_j(t)}{\tau_D} 
- \frac{1}{\tau_F} \max\left(1,\ \frac{h_f}{h}\right)\,
\max\left(0, 1- \frac{h}{h_f}\right)^\beta\, v_j
+ \lambda\,  v_j(t)\, \frac{v_{j+1}(t)-v_j(t)}{x_{j+1}(t)-x_j(t)}.
\label{eq:50} 
\end{eqnarray}
}
\subsubsection{Continuous flow limit}
%
This discrete model allow us to derive a continuous model by a scaling process
in both space and time. Let us introduce the density of debris denoted $\rho$.
The local density can be linked to inter-debris distance according to the formula
\begin{equation}
(\rho_d)(x_j(t)) = \frac{\rho^0\, \ell}{x_{j+1}(t)-x_j(t)}
\label{eq:51}
\end{equation}
where $\ell$ is a characteristic length of debris and $\rho^0$ is the maximum
debris density (without overlapping). In particular, from equation~\eqref{eq:49}
we have
\begin{equation}
\frac{d}{dt}(x_{j+1}-x_j) = \frac{v_{j+1}-v_j}{x_{j+1}-x_j}\, 
\left( x_{j+1}-x_j\right).
\label{eq:52}
\end{equation}
Applying the scaling to the expression~\eqref{eq:52} leads to a
continuous medium equation of conservation of the number of debris
\begin{equation}
\frac{D (1/\rho)}{Dt} = \frac{1}{\rho} \frac{\partial v}{\partial x}
\label{eq:53}
\end{equation}
where $\rho$ and $v$ are now functions of both space and time and
$D_t = \partial_t + v\partial_x$ is the Lagrangian derivative.
Equation~\eqref{eq:53} can be written in conservative form as
\begin{equation}
\partial_t \rho + \partial_x (\rho v) = 0.
\label{eq:54}
\end{equation}
The scaling applied to the expression~\eqref{eq:50} leads to the following
partial differential momentum-like equation
\begin{equation}
\rho D_t v - \lambda \rho v\, \frac{\partial v}{\partial x} = \frac{u-v}{\tau_D}
- \frac{1}{\tau_F}\max\left(1,\frac{h_f}{f}\right)
\max\left(0, 1-\frac{h}{h_f}\right)^\beta \, v,
\label{eq:55}
\end{equation}
or in Eulerian description
\begin{equation}
\partial_t (\rho v) + \partial_x(\rho v^2) - \lambda \rho v\, \frac{\partial v}{\partial x} = \frac{u-v}{\tau_D}
- \frac{1}{\tau_F}\max\left(1,\frac{h_f}{f}\right)
\max\left(0, 1-\frac{h}{h_f}\right)^\beta \, v.
\label{eq:56}
\end{equation}
In all what follows, for simplicity we will denote by $S=S(h,u,v)$ the right hand side
of equation~\eqref{eq:56} and will be referred to as the source term. 
For smooth solutions, equation~\eqref{eq:56} can be rewritten under the conservative
form
\begin{equation}
\partial_t v + \partial_x \left((1-\lambda)\frac{v^2}{2}\right) = \frac{S}{\rho}
\label{eq:57}
\end{equation}
that resembles a Burgers-like equation.
\subsubsection{Why velocity interaction terms are necessary~?}
%
In this section, we discuss the important role of the interaction term 
\[
I = - \lambda \rho v\, \frac{\partial v}{\partial x}.
\]
to get a well-posed mathematical problem. For simplicity purpose let us consider the particular case~$S=0$ (no drag and no friction forces). The system of partial
differential equation reads
\begin{eqnarray*}
&& \partial_t \rho + \partial_x(\rho v) = 0, \\ [1.1ex]
&& \partial_t v + (1-\lambda)\partial_x(v^2/2) = 0.
\end{eqnarray*}
As soon as $\lambda\neq 1$ and $v\neq 0$, the system is clearly hyperbolic with two
distinct eigenvalues $\lambda_1=(1-\lambda) v$ and $\lambda_2=v$.

Let us consider now the case $\lambda=0$ (non-existent interaction term). Then the system is
\begin{eqnarray*}
&& \partial_t \rho + \partial_x(\rho v) = 0, \\ [1.1ex]
&& \partial_t (\rho v) + \partial_x( \rho v^2) = 0.
\end{eqnarray*}
It is strictly equivalent to the so-called ``pressureless Euler equations''
(see for example Bouchut~\cite{bouchut}). This
system is not hyperbolic due to the unique eigenvalue $\lambda_1=v$ of multiplicity~2.
Actually this system can have weak measure solutions with the appearance of
what is called delta-shocks that are nothing else but Dirac measures. Indeed,
the second equation can be rewritten in $v$-variable as an inviscid Burgers equations, which is autonomous and can develop discontinuities in velocity. 
At discontinuity lines, the mass conservation equation says that there is
mass concentration. From the point of view of debris model, this would
represents Dirac concentration of debris, what is not really realistic from 
a physical point of view. This model would be the lowest-regularity description
of debris concentration.

In this sense, the interaction term acts as a regularization term that ``smooths''
debris concentration waves. Notice that in the particular case $\lambda=1$, the
interaction term kills the convective term  and we
get the interesting (conservative) simple model
\begin{eqnarray*}
&& \partial_t \rho + \partial_x (\rho v) = 0, \\ [1.1ex]
&& \partial_t v = \frac{S}{\rho}.
\end{eqnarray*}
Remark that the second equation can also be written
\[
\partial_t (\rho v) + v\, \partial_x(\rho v) = S,
\]
so that we have the balance law in conservation form
\[
\partial_t (\rho^2 v) + \partial_x(\rho^2 v) = S.
\]
\subsubsection{Extension to two-dimensional empirical debris model}
%
For two-dimensional problems, the models of drag and friction effects are
kept unchanged, except that they are now vector-valued~:
\[
\bS = \frac{\bu-\bv}{\tau_D}
- \frac{1}{\tau_F}\max\left(1,\frac{h_f}{f}\right)
\max\left(0, 1-\frac{h}{h_f}\right)^\beta \, \bv.
\]
For the modeling of debris interaction, the one-dimensional
debris-following approach is no more available. We simply empirically extend
the formula found in the 1D case, and propose to use
\begin{equation}
I = -\lambda\, \rho(\nabla\cdot\bv)\, \bv.
\label{eq:58}
\end{equation}
This empirical extension can be justified by the fact that $\nabla\cdot \bv$
represent a compressible factor that is positive for expansion conditions
and negative for compressive configuration. By this way, the regularizing interaction term acts in the opposite direction of 
$\mathop{sgn}(\nabla\cdot\bv)\bv$. \medskip

We get the following system of PDEs:
\begin{eqnarray}
&& \partial_t \rho + \nabla\cdot(\rho\bv) = 0, 
\label{eq:59} \\ [1.1ex]
&& \partial_t (\rho\bv) + \nabla\cdot(\rho \bv\otimes\bv) 
- \lambda \, \rho(\nabla\cdot\bv)\, \bv = \bS.
\label{eq:60}
\end{eqnarray}
For smooth solutions, the second equation can be written in velocity variable as
\[
\partial_t \bv + \bv\cdot\nabla\bv - \lambda \, (\nabla\cdot\bv)\, \bv = 
\frac{\bS}{\rho}.
\]
Unfortunately, unlike the one-dimensional case, the system cannot be written
in conservative form, except in the case $\lambda=1$. For $\lambda=1$, it is
easy to check that the momentum variable $\rho\bv$ is solution of the
simple transport-reaction equation
\begin{equation}
\partial_t (\rho\bv) + \bv\cdot\nabla(\rho\bv) = \bS.
\label{eq:61}
\end{equation}
Combining it with the continuity equation, we get the 
balance law in conservative form
\begin{equation}
\partial_t (\rho^2\bv) + \nabla\cdot(\rho^2\bv\otimes\bv) = \rho\bS.
\label{eq:62}
\end{equation}
Denoting $\bv=(v,w)$, the homogeneous part of the system can be written in 
quasi-linear form 
\begin{equation}
\partial_t \bv + A_x\, \partial_x \bv + A_y\, \partial_y \bv = 0
\label{eq:63}
\end{equation}
with 
\[
A_x =  \begin{pmatrix}0 & 0 \\ -w & v \end{pmatrix} ,\quad 
A_y = \begin{pmatrix}w & -v \\ 0 & 0 \end{pmatrix}.
\]
For any unit vector $\bm{\nu}=(\nu_x,\nu_y)$, we have
\[
A_\nu\ :=\ A_x\nu_x + A_y \nu_y = \begin{pmatrix} w\nu_y & -v\nu_y \\
-w \nu_x & v \nu_x \end{pmatrix}.
\]
It is clear that $\mathop{tr}(A_\nu)=\bv\cdot\bm{\nu}$ and $\mathop{det}(A_\nu)=0$
so that the eigenvalues of $A_\nu$ are $0$ and $\bv\cdot\bm{\nu}$.
\subsection{Two-way shallow water-debris coupling}
%
Accumulation of debris can create debris hills and thus act on the
global flooding flow. In that case, we have a two-way shallow water
equations-debris dynamics coupling. One can add to the ground topography $z$
a rising due to the debris, which is proportional to the debris density
(with factor $\mu<0$ in the next equations). Thus, the coupled system
is
\begin{eqnarray}
&& \partial_t h + \nabla\cdot(h\bu) = 0, \label{eq:64}\\ [1.1ex]
&& \partial_t (h\bu) + \nabla\cdot(h\bu\otimes \bu) + 
\nabla\cdot\left( g \frac{h^2}{2}\right) = -g\,h\,\nabla(z+\mu\rho), 
\label{eq:65} \\ [1.1ex]
&& \partial_t \rho + \nabla\cdot (\rho \bv) = 0, \label{eq:66} \\ [1.1ex]
&& \partial_t (\rho\bv) + \nabla\cdot(\rho \bv\otimes\bv) 
- \lambda \, \rho(\nabla\cdot\bv)\, \bv = \bS.
\label{eq:67}
\end{eqnarray}
\subsection{Visualization issues}
%
Visualization of debris dynamics added to the water flooding may be
a complementary tool for decision support and risk analysis.
One can visualize the debris density field but it is not so easy to 
merge both water surface and debris density. Another practical
and humanly comprehensible way is to visualize debris by debris
particles themselves. In some sense we may go back to a Lagrangian description
from the Eulerian computations, for visualization purposes only. While
both $\rho$ and $\bv$ are computed. We define a set of debris particles
can can be initially instantiated according to the initial debris density.
Then we compute the trajectories of the set of debris particles
\begin{eqnarray}
&& \Dot x_j = \bv(x_j(t),t), \quad t>0, \\ [1.1ex]
&& x_j(0) = x_j^0.
\label{eq:p}
\end{eqnarray}
From a user perspective, this will give him an overview of debris dynamics,
zones of concentration or accumulation, particular pathways, assessment
of human safety and risks for infrastructures. 
\subsection{Candidate damage functions}
%
Damages are caused by stresses and forces exerted on structures by both
water and debris. Considering debris, the accumulated momentum of debris
at a given point point $x$ and at time $T$ is
\begin{equation}
\bm{D}(x,T) = \int_0^T \rho \bv(x,t)\, dx.
\label{eq:68}
\end{equation}
Of course, debris density $\rho(.,T)$ itself is another indicator of 
possible damage. 
The case for example of density of debris releases into the sea is a good
indicator of risk assessment for future coastal shipping activities.
\section{Selected simple empirical debris model and associated numerical scheme}
%
As a pioneer work, we decide to use the simplest debris model, i.e. considering
$\lambda=1$ without two-way coupling. For numerical expectations, we rather
use the conservative formulation of debris model, that we recall here again:
\begin{eqnarray*}
&& \partial_t \rho + \nabla\cdot (\rho \bv) = 0, \\ [1.1ex]
&& \partial_t (\rho^2\bv) + \nabla\cdot(\rho^2 \bv\otimes\bv) 
 = \rho\bS.
\end{eqnarray*}
For numerical discretization, one can use a Lagrange-flux scheme to solve the
convective part of the system. Let us now focus on the numerical treatment of the source term in the system. At each time step~$\Delta t^n$ and each grid point $x$,
we have to solve the differential problem
\begin{eqnarray}
&& \frac{d\bv}{dt} = \frac{\bu-\bv}{\tau_D}
- \frac{1}{\tau_F}\max\left(1,\frac{h_f}{h}\right)
\max\left(0, 1-\frac{h}{h_f}\right)^\beta \, \bv, \label{eq:69}\\ [1.1ex]
&& \bv(t^n) =  \bv^n. \label{eq:70}
\end{eqnarray}
For envisaged applications, both relaxation times $\tau_D$ and $\tau_F$ can be rather
small compared to the time step, leading to stiff source terms. For that reason,
it appears important to use at least semi-implicit time integration schemes.

By freezing up exogenous terms with values taken at time $t^n$, we have the
following linear differential equation
\[
\frac{d\bv}{dt} = \frac{\bu^n}{\tau_D}
- \left[ \frac{1}{\tau_D} + \frac{1}{\tau_F}\max\left(1,\frac{h_f}{h^n}\right)
\max\left(0, 1-\frac{h^n}{h_f}\right)^\beta \right] \, \bv.
\]
We get the solution at time $t^{n+1}=t^n+\Delta t^n$
\begin{equation}
\bv^{n+1} = \bv(\Delta t) = 
\bv^n \,  e^{-\eta^n\Delta t^n} + \left(1-e^{-\eta^n\Delta t^n} \right) 
\frac{\bu^n}{\tau_D\, \eta^n}
\label{eq:71}
\end{equation}
with
\[
\tau^n = \frac{1}{\tau_D} + \frac{1}{\tau_F}\max\left(1,\frac{h_f}{h^n}\right)
\max\left(0, 1-\frac{h^n}{h_f}\right)^\beta.
\]
\section{Numerical experiments}
%
To show the effectiveness of the numerical model, we consider a two-dimensional case
defined on the rectangle spatial domain $[0,2]\times [0,1]$. 
The topography/bathymetry is that one of the figure~\ref{fig:4} below.
The tidal wave is generated by an artificial discontinuous water elevation.
The flow dynamics is made a coastal tidal wave that submerges all the surface,
then a pull-back wave starts and goes back to the sea. The three obstacles
also create backward waves.  For the debris, we initially define
a nonzero uniform debris density $\rho=2$ in the sub-zone $[0.7,\, 1]\times[0,1]$.
\begin{figure}[h!]
\centering\includegraphics[width=0.9\textwidth]{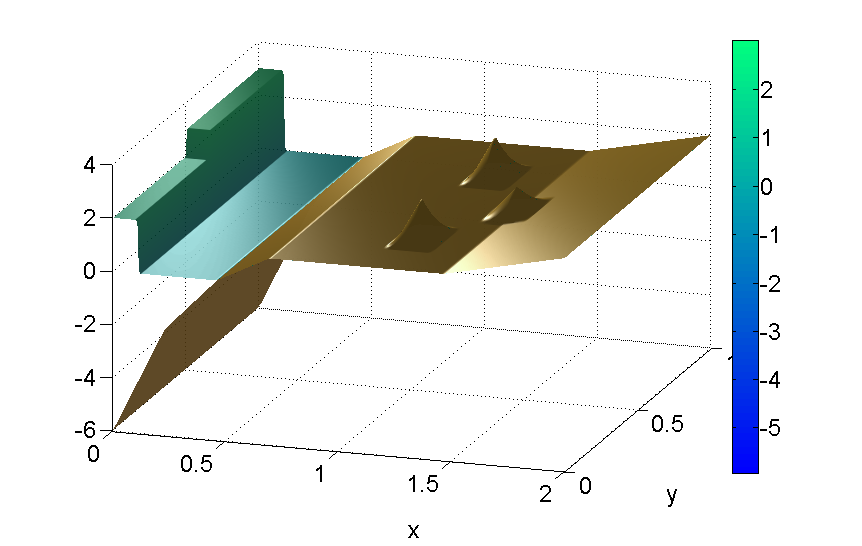}
\caption{Configuration at initial time with topography/bathymetry (in brown color)
and water depth (green-blue color).}\label{fig:4}
\end{figure}

Numerical results showing the sequence of the flow dynamics is given in figure~\ref{fig:5}. One can observe first the compression of the debris caused by 
the tidal wave, then the transportation of the debris that either land on the 
borders of the three hills or concentrate between the hills. We also compute
a cumulative damage function as the time integral of the norm of the 
debris momentum:
\[
D(.,t) = \int_0^t \rho |\bv|(.,s)\, ds.
\]  
One can observe that most damages are located as intuitively guessed
between the obstacles where the flow rate is maximal. 
This promising qualitative result let us think that our debris model
is appropriate for damage assessment caused by debris. 
Of course, for a quantitative computation, we would need to calibrate
debris flotation heights and drag coefficients from real measurements or 
laboratory experiments.
\begin{figure}[h!]
\begin{center}
\includegraphics[width=0.32\textwidth]{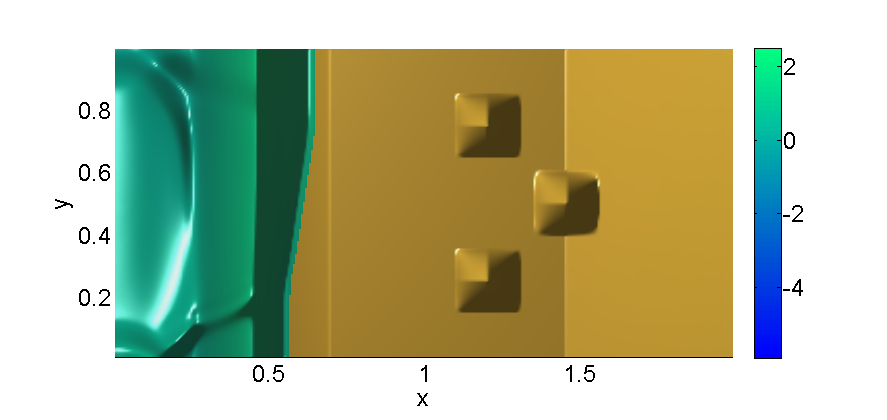}
\includegraphics[width=0.32\textwidth]{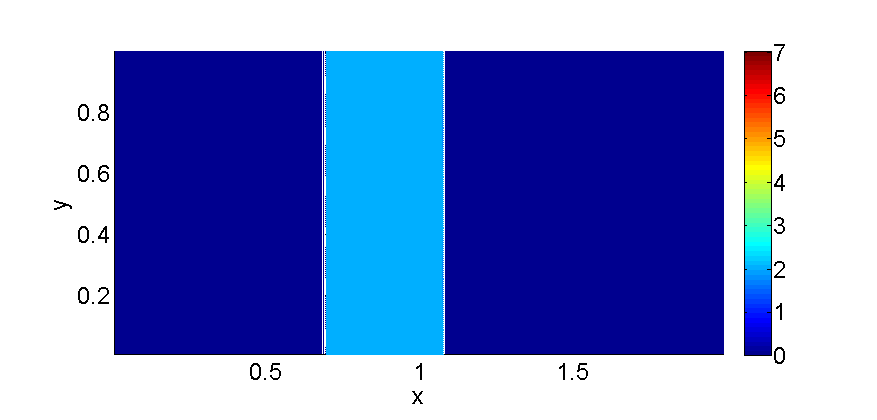}
\includegraphics[width=0.32\textwidth]{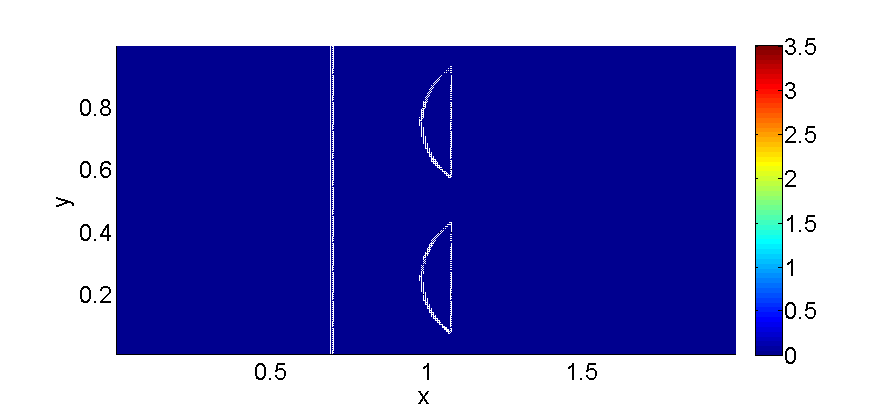}
%
\includegraphics[width=0.32\textwidth]{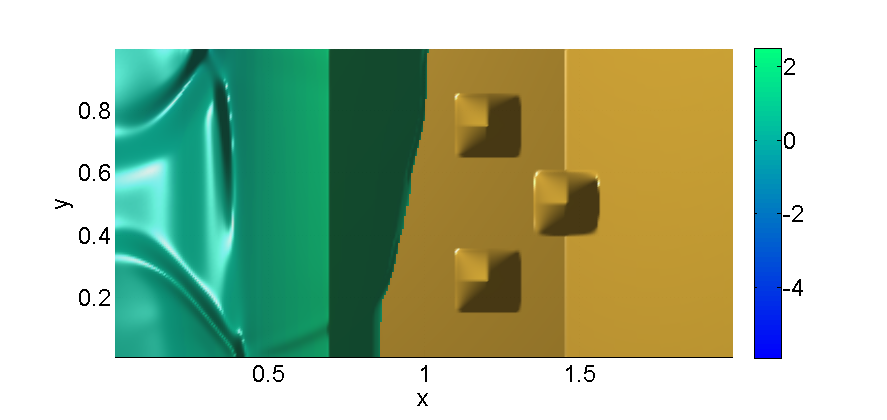}
\includegraphics[width=0.32\textwidth]{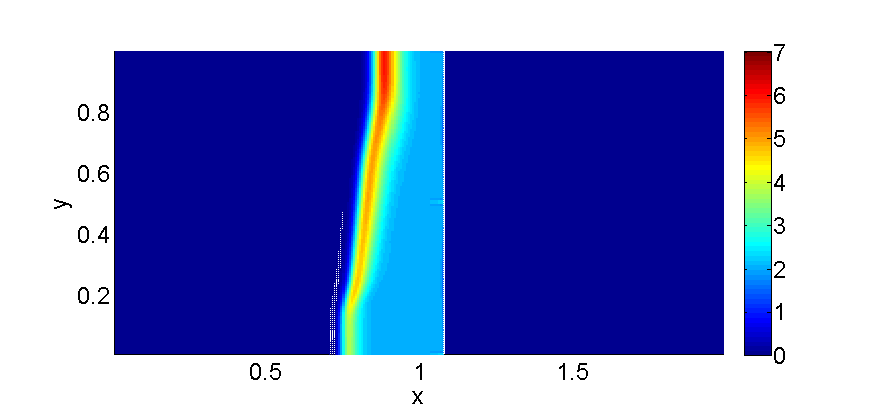}
\includegraphics[width=0.32\textwidth]{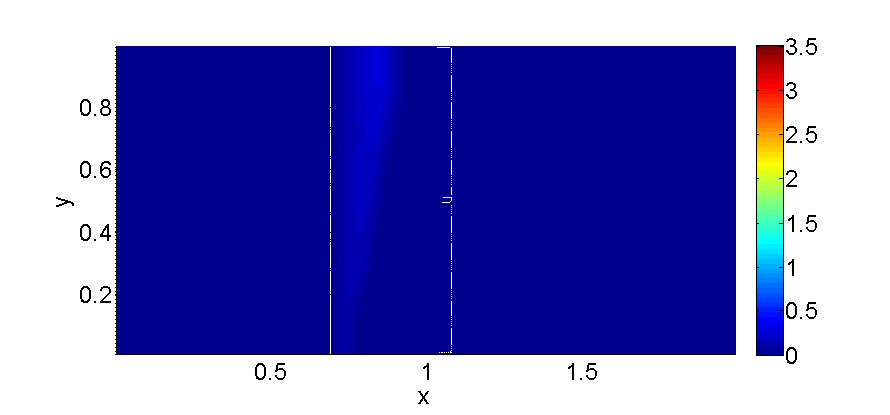}
\includegraphics[width=0.32\textwidth]{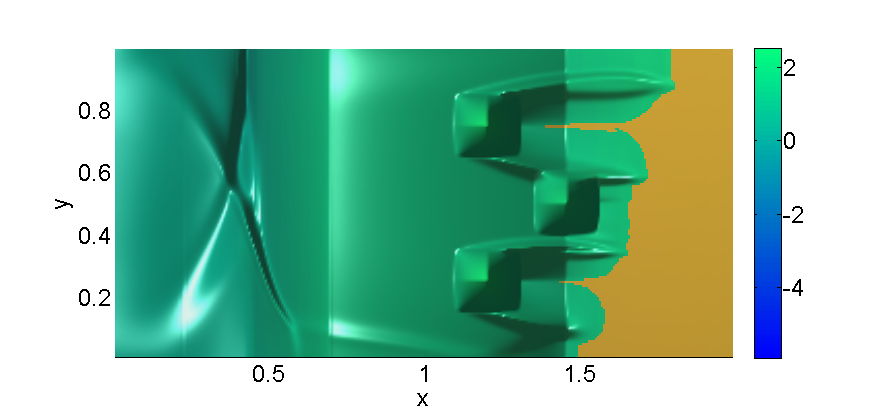}
\includegraphics[width=0.32\textwidth]{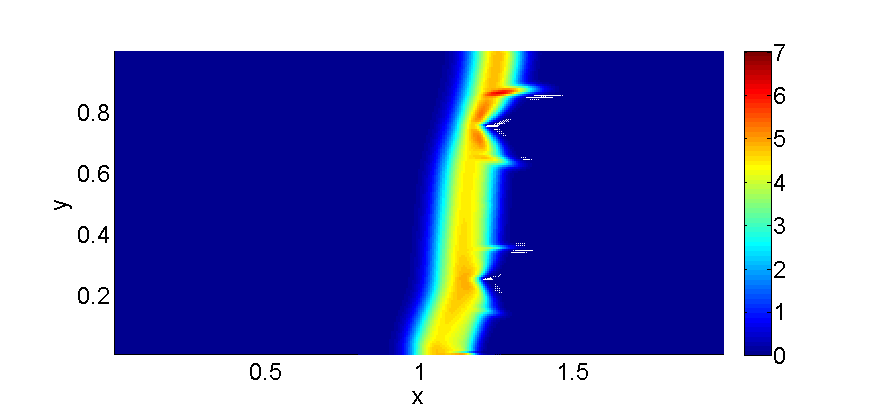}
\includegraphics[width=0.32\textwidth]{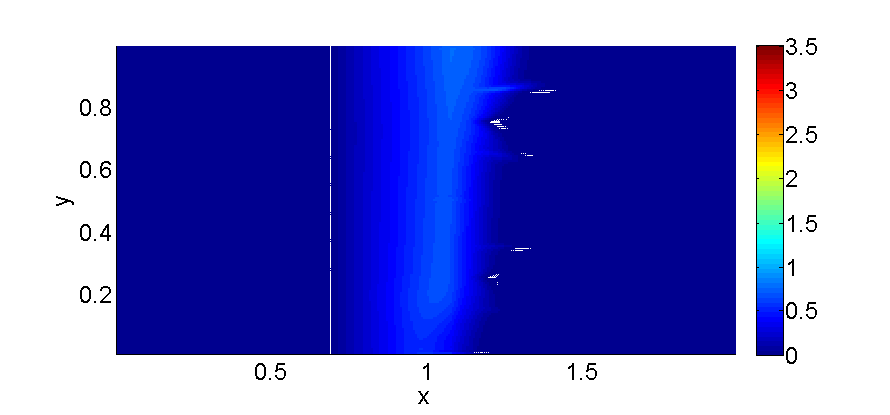}
\includegraphics[width=0.32\textwidth]{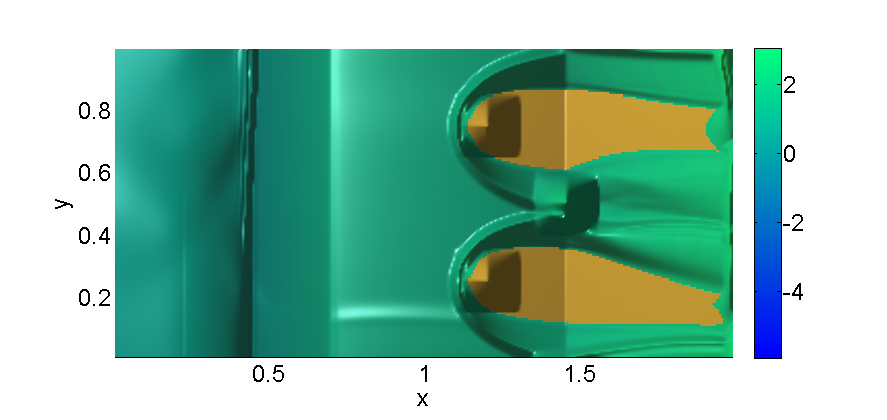}
\includegraphics[width=0.32\textwidth]{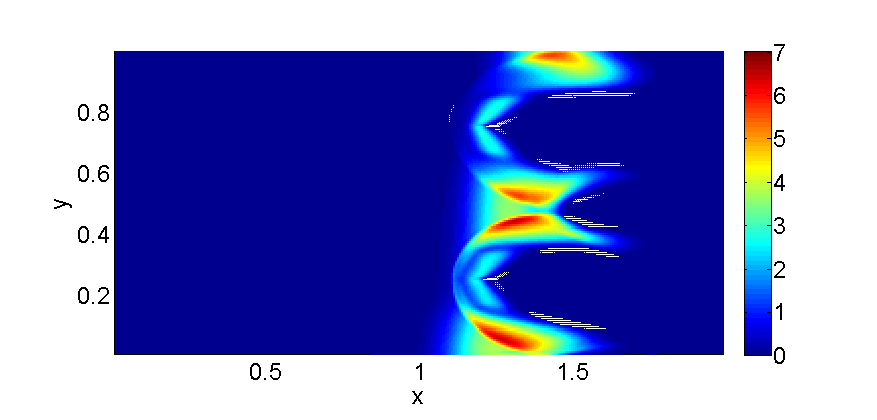}
\includegraphics[width=0.32\textwidth]{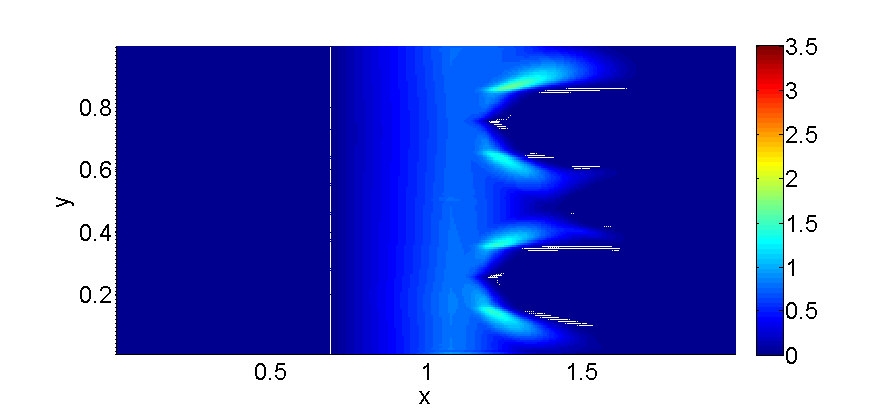}
\includegraphics[width=0.32\textwidth]{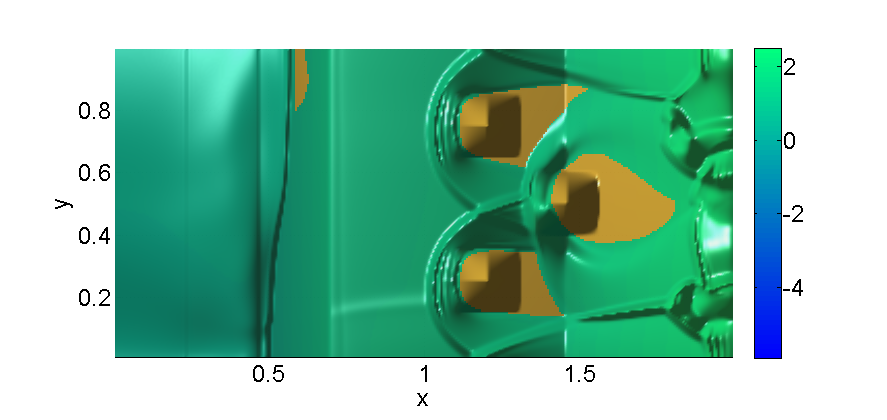}
\includegraphics[width=0.32\textwidth]{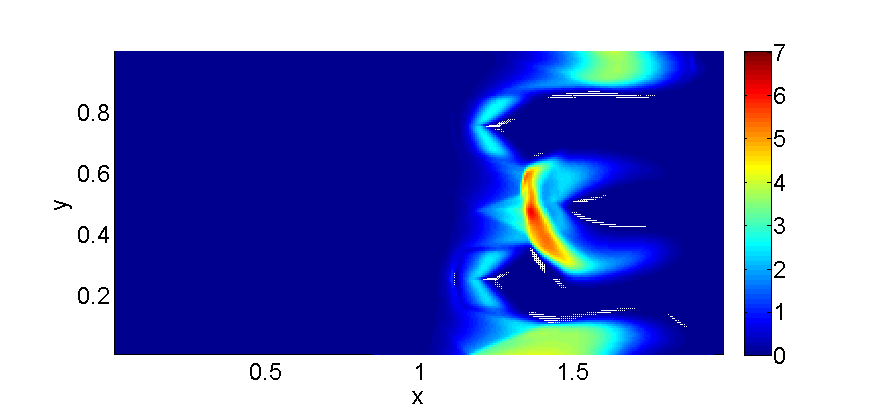}
\includegraphics[width=0.32\textwidth]{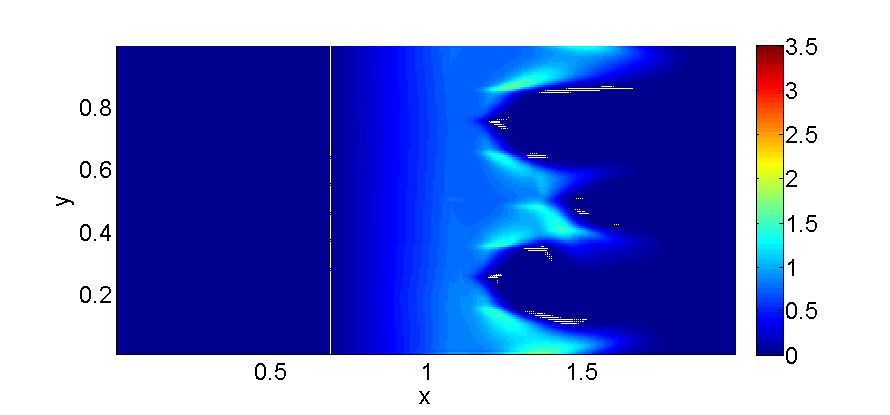}
\includegraphics[width=0.32\textwidth]{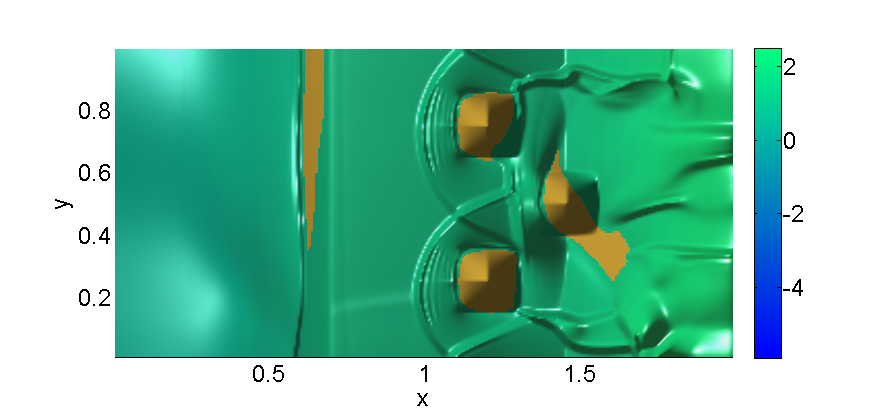}
\includegraphics[width=0.32\textwidth]{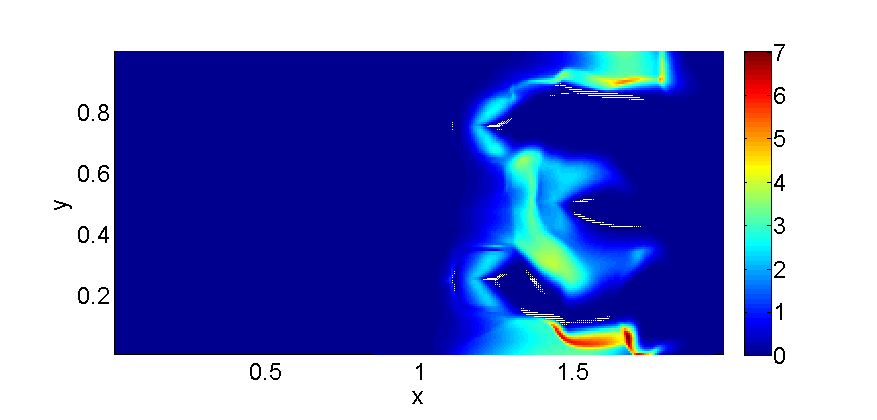}
\includegraphics[width=0.32\textwidth]{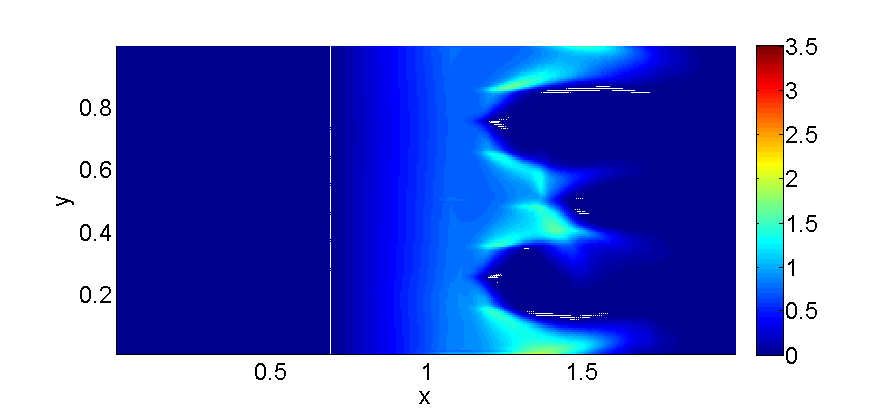}
\includegraphics[width=0.32\textwidth]{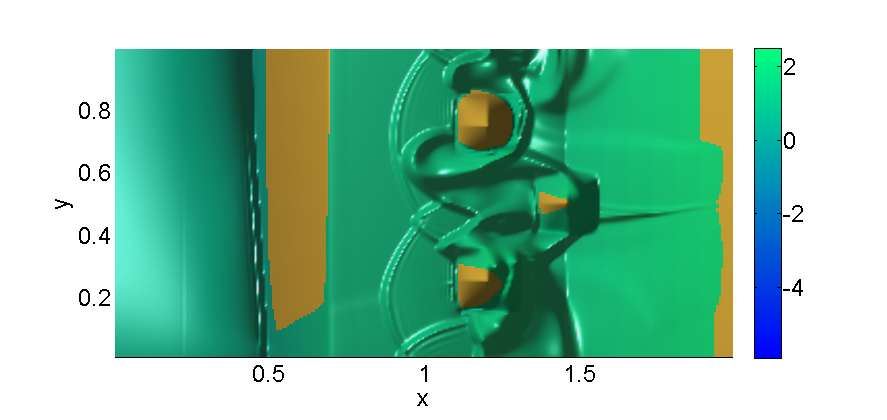}
\includegraphics[width=0.32\textwidth]{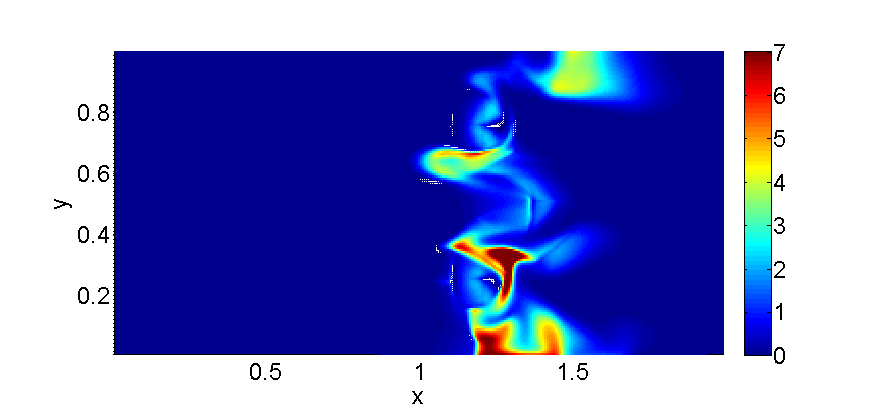}
\includegraphics[width=0.32\textwidth]{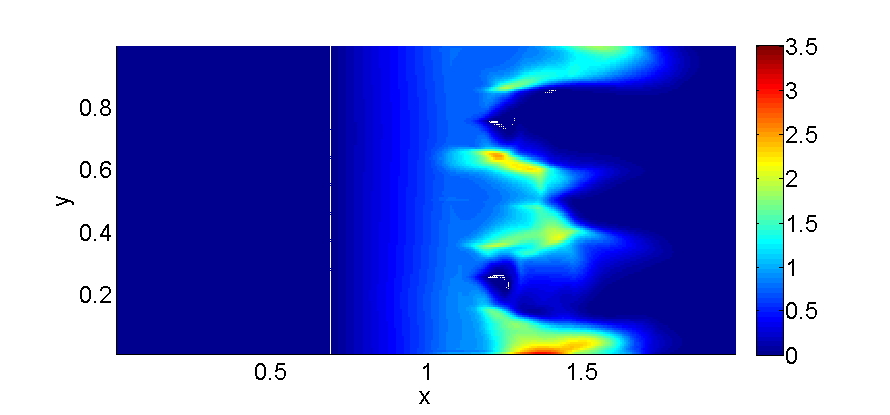}
\includegraphics[width=0.32\textwidth]{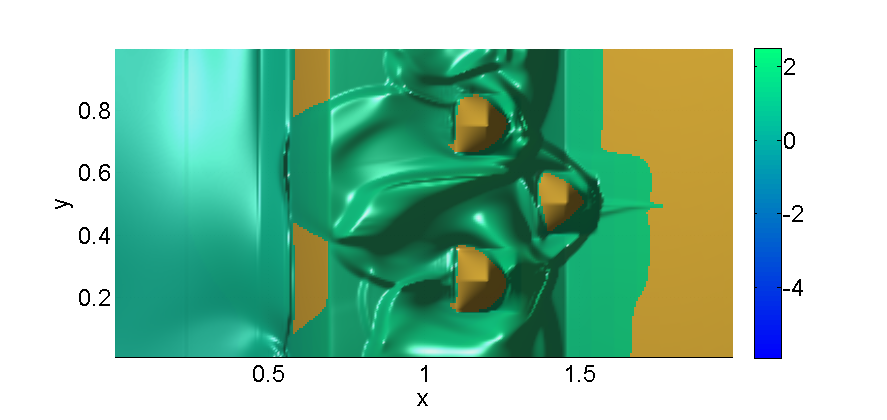}
\includegraphics[width=0.32\textwidth]{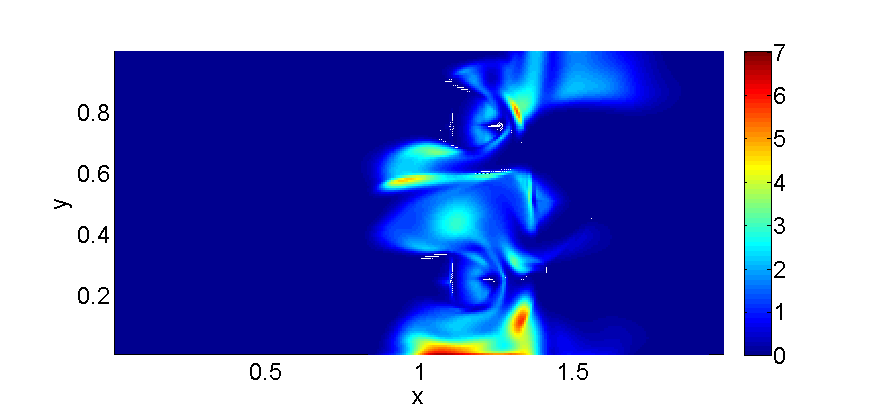}
\includegraphics[width=0.32\textwidth]{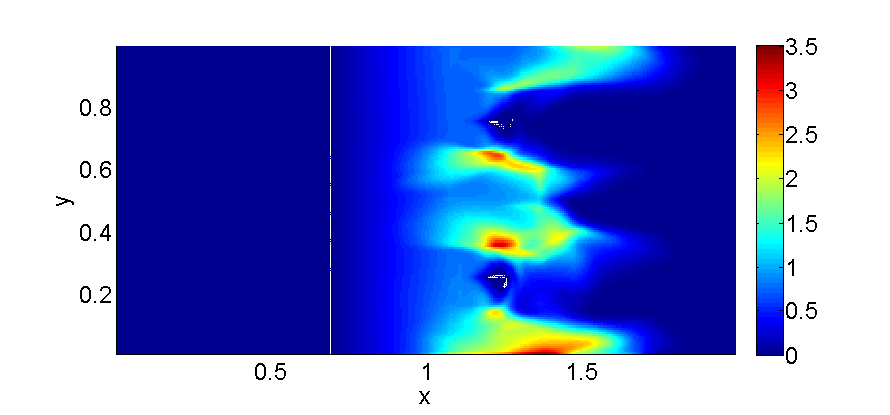}
\includegraphics[width=0.32\textwidth]{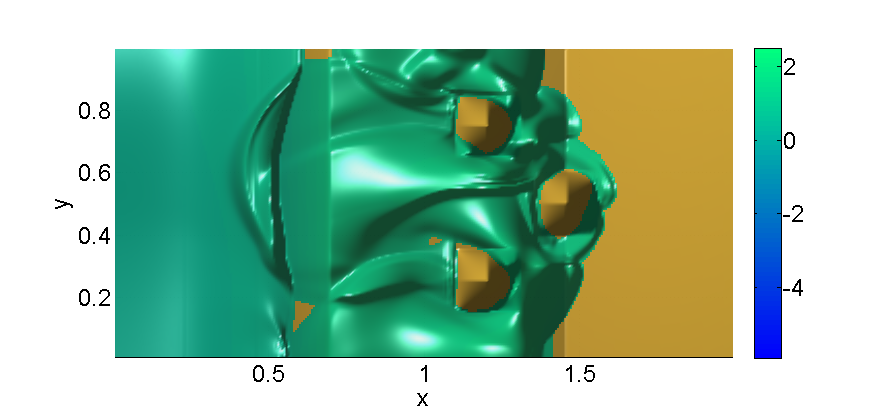}
\includegraphics[width=0.32\textwidth]{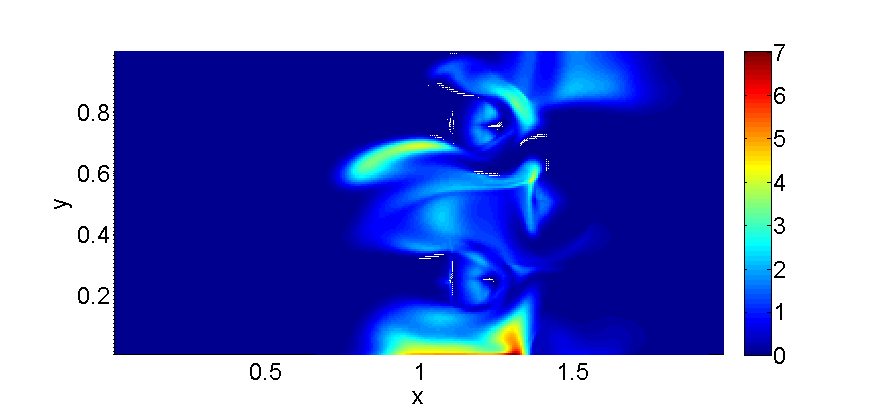}
\includegraphics[width=0.32\textwidth]{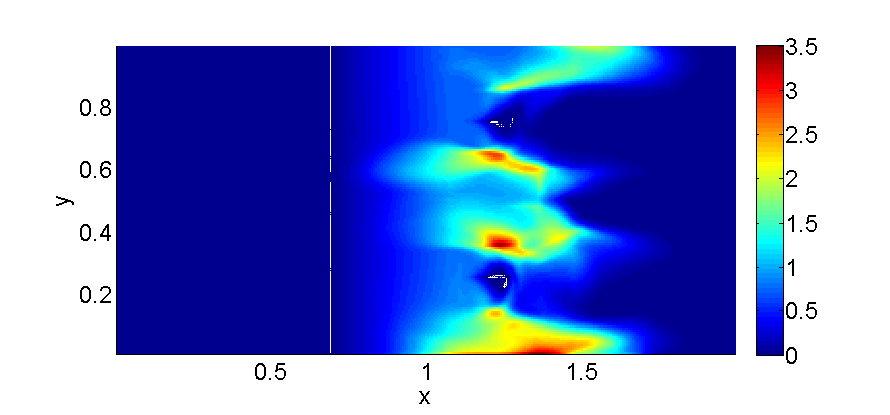}
\end{center}
\caption{Numerical experiment of tsunami flooding and debris dynamics with damage function.
From left to right: flooding flow, debris density and damage indicator. Top to bottom: at different successive times of the flooding.}\label{fig:5}
\end{figure}
\section{Outlook and next working program}
%
The next step is to introduce this debris model into a high-performance code like GPU-VOLNA \cite{dutykh,giles}, then use a computer design of experiment (DoCE) with damage analysis, sensitivity analysis and
uncertainty quantification. 

A way to reduced data dimensionality is to consider reduced-order models (ROM)
for spatially-defined damage functions like those proposed in the document.
One can notice from the numerical results that that damage function defined
by the time integral of debris momentum is a rather smooth function in space
so that model-order reduced can be envisaged. This approach will simplify
the generation of emulators for damage functions.

We will also try to propose original numerical methodologies but also think about suitable parallel software environments to process
the data and extract knowledge in this risk assessment framework.
\section*{Acknowledgments} 
%
I would like to warmly thank Professor Serge Guillas from the Department of Statistical Science of UCL who supported my candidature as UCL Big Data Institute invited Researcher. This work is also partly supported and granted by the french CNRS multidisciplinary program ``D\'efi Littoral'' 2015-2016.
%
%

%
\end{document}